\documentclass[10pt,twocolumn,letterpaper]{article}
\usepackage{cvpr}      
\usepackage{multirow}
\usepackage{adjustbox}
\usepackage{makecell}
\usepackage{subcaption}
\usepackage{graphicx}
\usepackage{caption}
\usepackage{adjustbox}
\usepackage{float}
\usepackage{amsthm}
\newtheorem{theorem}{Theorem}
\newtheorem{definition}{Definition}
\definecolor{cvprblue}{rgb}{0.21,0.49,0.74}
\usepackage[pagebackref,breaklinks,colorlinks,allcolors=cvprblue]{hyperref}

\def\paperID{10198} 
\def\confName{CVPR}
\def\confYear{2026}

\title{Communication-Efficient Serving for Video Diffusion Models with Latent Parallelism}

\author{
    Zhiyuan Wu$^{1}$ \quad 
    Shuai Wang$^{2}$ \quad 
    Li Chen$^{2}$ \quad 
    Kaihui Gao$^{2}$ \\
    Dan Li$^{1,2,\dagger}$ \quad 
    Yanyu Ren$^{1}$ \quad 
    Qiming Zhang$^{3}$ \quad 
    Yong Wang$^{3}$ \\
    \\
    $^{1}$Department of Computer Science and Technology, Tsinghua University \\
    $^{2}$Zhongguancun Laboratory \quad 
    $^{3}$ZTE Corporation \\
    {Email: wu-zy25@tsinghua.org.cn}
}

\begin{document}
\maketitle
\begin{abstract}
Video diffusion models (VDMs) perform attention computation over the 3D spatio-temporal domain. Compared to large language models (LLMs) processing 1D sequences, their memory consumption scales cubically, necessitating parallel serving across multiple GPUs. Traditional parallelism strategies partition the computational graph, requiring frequent high-dimensional activation transfers that create severe communication bottlenecks. To tackle this issue, we exploit the local spatio-temporal dependencies inherent in the diffusion denoising process and propose Latent Parallelism (LP), the first parallelism strategy tailored for VDM serving. \textcolor{black}{LP decomposes the global denoising problem into parallelizable sub-problems by dynamically rotating the partitioning dimensions (temporal, height, and width) within the compact latent space across diffusion timesteps, substantially reducing the communication overhead compared to prevailing parallelism strategies.} To ensure generation quality, we design a patch-aligned overlapping partition strategy that matches partition boundaries with visual patches and a position-aware latent reconstruction mechanism for smooth stitching. Experiments on three benchmarks demonstrate that LP reduces communication overhead by up to 97\% over baseline methods while maintaining comparable generation quality.
As a non-intrusive plug-in paradigm, LP can be seamlessly integrated with existing parallelism strategies, enabling efficient and scalable video generation services.
\end{abstract}

\begin{figure}[t]
    \centering
    \begin{minipage}{\linewidth}
        \centering
        \includegraphics[width=1.0\textwidth]{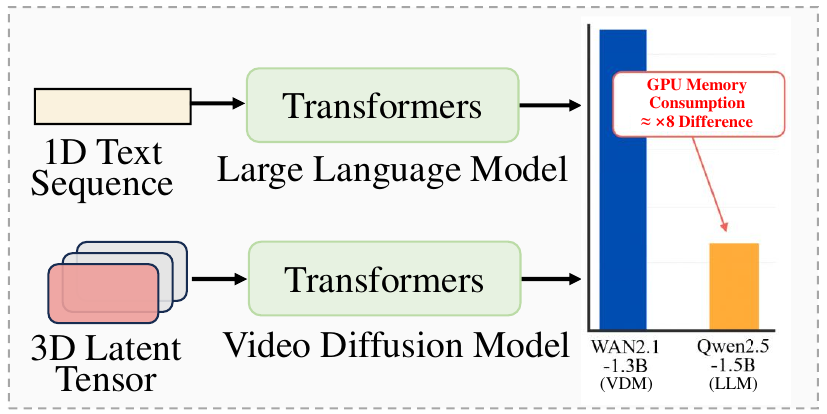}
        \vspace{-18pt}
        \caption{Comparison of VDM and LLM in terms of GPU memory consumption.}
        \label{vdm-vs-llm}
    \end{minipage}\par\medskip
\end{figure}

\section{Introduction}
\label{sec:introduction}
In recent years, video diffusion models (VDMs) have emerged as a transformative paradigm in AI-Generated Content (AIGC) \cite{xing2024survey,ma2025efficient}. State-of-the-art VDMs such as Sora \cite{brooks2024video}, Veo \cite{google2025veo}, and WAN \cite{wan2025wan} can now generate high-fidelity videos spanning thousands of frames at 1080p resolution with remarkable temporal consistency. However, this capability comes at a prohibitive cost: VDMs face extremely high GPU memory demands, fundamentally different from those of large language models (LLMs). While LLMs process 1D text sequences, VDMs perform attention computation over a 3D latent domain consisting of temporal, width, and height dimensions, leading to cubic scaling of GPU memory consumption. As illustrated in Figure \ref{vdm-vs-llm}, generating merely a 5-second, 480p video using WAN2.1-1.3B \cite{wan2025wan} requires over 30GB of memory on an A6000 GPU. This memory footprint is approximately 8× larger than that of Qwen-2.5-1.5B \cite{qwen2,qwen2-5}, making video generation with VDMs challenging on a single memory-constrained GPU. \nocite{zet-paper1}

To overcome this memory barrier, parallel serving across multiple GPUs becomes necessary for serving VDMs, yet dedicated parallelism strategies remain largely unexplored. Consequently, researchers routinely adopt conventional LLM parallelism strategies for VDMs, which exhibit significant limitations for video generation workloads \cite{wan2-1-gpu,hunyuan-gpu,yangscalefusion}.
Data parallelism (DP), which replicates the model across multiple GPUs, fails to reduce the prohibitive per-request memory footprint \cite{pytorch_dataparallel_docs}. In contrast, model parallelism \cite{brakel2024model} and its various forms partition the computational graph, distributing memory load across GPUs. Specifically, naive model parallelism (NMP) assigns consecutive model layers to different GPUs \cite{native-model-parallel}, tensor parallelism (TP) splits individual operators across GPUs at intra-layer granularity \cite{shoeybi2019megatron}, and pipeline parallelism (PP) splits the model into stages and processes multiple micro-batches concurrently to improve throughput \cite{butler2024pipeinfer}. While effective for LLMs, these strategies encounter severe communication bottlenecks in VDMs due to the dimensional gap in activation tensors. 

Unlike LLMs that process relatively compact 1D token sequences, VDMs generate massive 3D spatio-temporal activations at each denoising timestep that must be transmitted across GPUs. 
Quantitatively, deploying WAN2.1-1.3B on 4 GPUs to generate an 81-frame video results in over 90GB of total communication overhead per inference for NMP and PP, as entire high-dimensional activation tensors are transferred between the boundaries of diffusion transformer (DiT) blocks. While TP reduces per-operation communication by partitioning along the intermediate tensor dimension, it requires frequent collective operations within DiT blocks, leading to cumulative network traffic that remains substantial across the entire denoising process.
This issue is exacerbated in real-world deployments where cost-effective PCIe interconnects are widely used, despite offering only a fraction of the bandwidth provided by NVLink \cite{li2019evaluating,hu2025demystifying}. In such bandwidth-limited service clusters, cross-GPU communication becomes the critical bottleneck that fundamentally impedes the practical deployment of VDM serving systems. Therefore, designing a communication-efficient parallelism strategy for serving VDMs without compromising generation quality remains an open challenge.

To address this challenge, we re-examine parallelism strategies from first principles, considering the intrinsic characteristics of VDMs. \textcolor{black}{Unlike prevailing strategies (NMP, TP, PP) that partition the computational graph and inevitably transmit high-dimensional activation tensors across GPUs, we propose a new perspective that shifts the target of parallelization: applying divide-and-conquer to the latent tensor of individual requests rather than partitioning the model.} This paradigm shift is enabled by two key properties of video generation. 
First, diffusion denoising at each spatio-temporal location primarily depends on its local neighborhood, allowing the 3D latent space to be partitioned into sub-regions that can be denoised independently across GPUs before being reconstructed.
Second, the latent tensor is relatively lightweight, being an order of magnitude smaller than model parameters and intermediate activations. 
These properties enable a communication-efficient parallelism paradigm for VDMs: instead of repeatedly transferring massive activations across DiT block boundaries, latent-domain parallelization partitions the spatio-temporal space and only exchanges compact representations, drastically reducing communication overhead. \nocite{zet-paper2}

In this paper, we propose Latent Parallelism (LP), a parallel serving strategy specifically designed for VDMs. LP employs a dynamic rotating partition method that alternately decomposes the video latent space along temporal, height, and width dimensions across different denoising timesteps, transforming the global denoising task into parallelizable local sub-problems across GPUs. \textcolor{black}{This design allows each local region to access the full spatio-temporal context of the entire video through rotation cycles, thereby maintaining global consistency throughout the diffusion process.}
To preserve generation quality, we develop two complementary techniques. First, we propose a patch-aligned overlapping partition strategy that matches partition boundaries with the internal visual patches of VDMs, preserving latent space feature integrity while ensuring high-quality local denoising. Second, we design a position-aware latent reconstruction mechanism that adaptively weights overlapping regions based on their positions, enabling seamless stitching and eliminating boundary artifacts.
\textit{\textbf{To the best of our knowledge, LP is the first parallelism strategy tailored for VDM serving.}} As a non-intrusive paradigm, LP can be seamlessly integrated with existing parallelism strategies (DP, NMP, TP, PP) and feature cache reuse methods \cite{liu2025timestep,ma2025model,liu2025reusing,ma2025magcache}, making it a promising plug-in component for efficient VDM serving systems.

The contributions of this paper are as follows:

\begin{itemize}
\item We propose LP, the first parallelism strategy tailored for VDM serving that dynamically partitions the latent space across GPUs, achieving communication-efficient parallel serving while maintaining generation quality.

\item We design two complementary techniques for LP to preserve generation quality: patch-aligned overlapping partitioning that matches boundaries with the internal visual patches of VDM, and position-aware latent reconstruction that eliminates boundary artifacts during stitching.

\item We provide theoretical analysis on the feasibility of LP in ensuring video quality and characterize its communication overhead, demonstrating its significant advantages over conventional parallelism strategies.

\item We conduct experiments on EvalCrafter \cite{liu2024evalcrafter}, T2V-CompBench \cite{sun2025t2v}, and VBench \cite{huang2024vbench} benchmarks. Results show that LP reduces communication overhead by up to 97\% over conventional parallelism strategies while maintaining comparable video generation quality.
\end{itemize}

\section{Related Work}
\subsection{Video Diffusion Model}
Video diffusion models have evolved from pixel-space to latent-space generation. Early approaches inflated 2D U-Nets with temporal attention layers \cite{blattmann2023stable,chen2023videocrafter1}, with VideoCrafter \cite{chen2023videocrafter1} employing dual cross-attention over full CLIP patch tokens for content preservation. Recent research adopts the scalable DiT blocks \cite{peebles2023scalable} as the VDM backbone. These works, exemplified by Sora \cite{brooks2024video}, WAN \cite{wan2025wan}, and HunyuanVideo \cite{kong2024hunyuanvideo}, have demonstrated significant improvements in video generation quality and temporal coherence through transformer-based architectures \cite{brooks2024video,wan2025wan,kong2024hunyuanvideo}. More recent VDMs like WAN2.2 \cite{Wan2-2} have incorporated mixture-of-experts (MoE) architecture, dividing denoising stages between specialized experts to increase capacity at constant computational cost. However, these methods jointly process spatio-temporal dimensions throughout the video, causing cubic scaling of the attention matrix and excessive GPU memory usage.

\nocite{zet-paper3,servellm}
\subsection{VDM Inference Acceleration}
Recent research accelerates VDM inference via two primary approaches: dynamic feature caching and inference system optimization.
Specifically, dynamic feature caching methods reduce redundant computations by adaptively reusing intermediate outputs \cite{liu2025timestep,ma2025model,liu2025reusing,ma2025magcache,zhou2025less}. 
As an early approach, TeaCache \cite{liu2025timestep} leverages timestep embeddings to estimate output differences as a reuse criterion, but requires extensive prompt-specific calibration. More recent works leverage robust, runtime-adaptive criteria derived from internal model dynamics, such as the unified magnitude law \cite{ma2025magcache} or the transformation rate \cite{zhou2025less}. A semantic-aware alternative, ProfilingDiT \cite{ma2025model}, selectively caches static background blocks while preserving full computation for dynamic elements. Another line of research focuses on inference system optimization. ScaleFusion \cite{yangscalefusion}, as a representative, optimizes DiT inference by overlapping computation and cross-machine communication through intra- and inter-layer scheduling. However, these methods either optimize the inference process of VDMs or adhere to existing parallelism strategies, leaving novel parallelism perspectives that could reshape the inference paradigm of VDM largely unexplored.

\begin{figure}[t]
    \centering
    \begin{minipage}{\linewidth}
        \centering
        \includegraphics[width=1.0\textwidth]{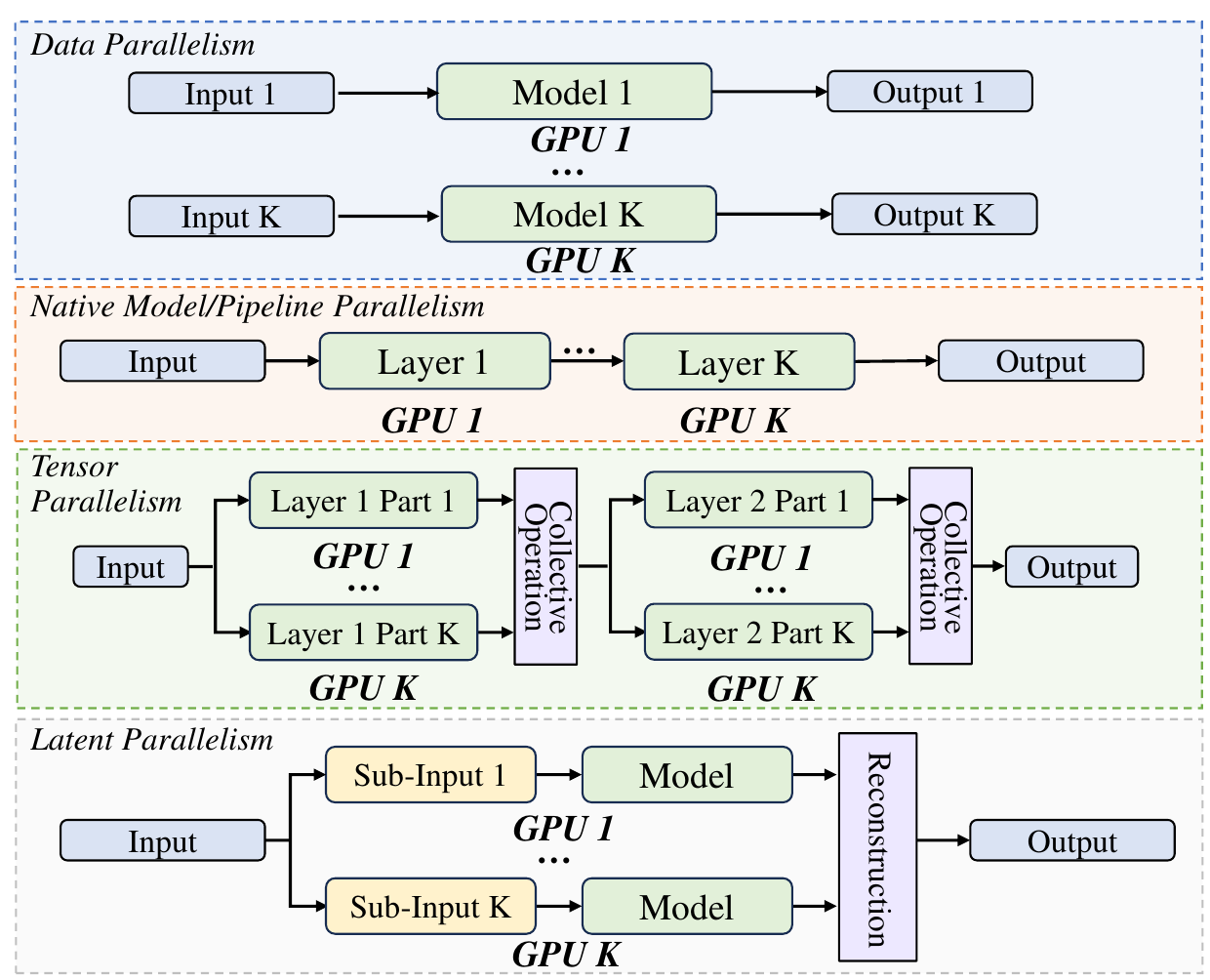}
        \vspace{-18pt}
        \caption{Comparison of different parallelism strategies.}
        \label{parallel-serving}
    \end{minipage}\par\medskip
\end{figure}

\begin{figure*}[t]
    \centering
    \begin{minipage}{\linewidth}
        \centering
        \includegraphics[width=0.95\textwidth]{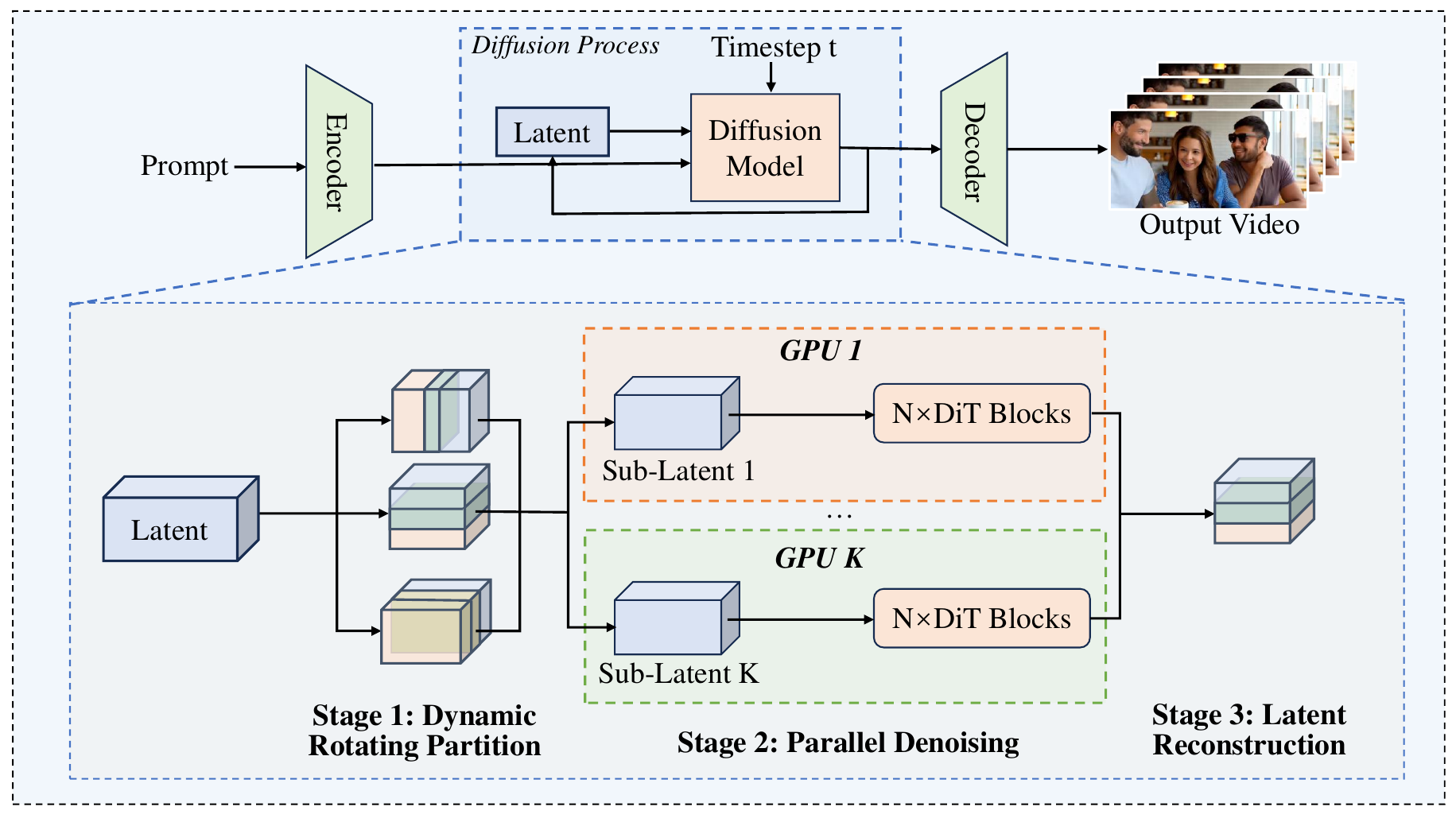}
        \vspace{-10pt}
        \caption{Workflow of LP. At each denoising timestep $t$, the global latent is decomposed through (1) dynamic rotating partition, processed via (2) parallel denoising on independent GPUs, and finally unified by (3) latent reconstruction.}
        \label{framework}
    \end{minipage}\par\medskip
\end{figure*}

\section{Method}
\label{sec:method}

\subsection{Preliminaries}
\label{subsec:preliminaries}
VDMs generate high-quality videos from random noise through an iterative denoising process. Given an encoded text prompt $c$, the process begins with a randomly sampled latent tensor $z_T$ and progressively denoises it over $T$ timesteps to obtain the clean latent $z_0$, which is subsequently decoded into the final video. The denoising process at each timestep can be formalized as:
\begin{equation}
z_{t-1} = \mathcal{S}(z_{t}, \tilde{f}(z_{t}, t, c), t),
\end{equation}
where $t \in \{T,T-1,\ldots,1\}$ is the diffusion timestep, and $z_t$ represents the noisy latent tensor at timestep $t$. $\tilde{f}(\cdot)$ denotes the guided denoising network, and $\mathcal{S}(\cdot)$ is the sampling scheduler that computes the less noisy latent tensor $z_{t-1}$ from $z_t$ and the predicted noise $\tilde{f}(z_t, t, c)$. Typically, VDMs adopt the classifier-free guidance (CFG) \cite{ho2022classifier} framework, which operates by linearly combining conditional and unconditional noise predictions:
\begin{equation}
\tilde{f}(z_t, t, c) = f(z_t, t, \emptyset) + w \cdot (f(z_t, t, c) - f(z_t, t, \emptyset)),
\end{equation}
where $f(\cdot)$ denotes the denoising network with DiT backbone, $f(z_t, t, c)$ is the conditional prediction based on encoded prompt $c$, $f(z_t, t, \emptyset)$ is the unconditional prediction based on the embedding of the null prompt $\emptyset$, and $w$ is the guidance scale hyperparameter. \textcolor{black}{Internally, $f(\cdot)$ divides the latent tensor into non-overlapping spatio-temporal patches, which serve as the atomic processing units in DiT blocks.}

Operating on latent tensor $z_t$ across temporal, height, and width dimensions, the denoising network $f(\cdot)$ incurs prohibitive GPU memory demands that necessitate parallel serving across multiple GPUs. However, conventional parallelism strategies (NMP, TP, PP) introduce severe communication bottlenecks in this iterative denoising process. At each of the $T$ timesteps, massive intermediate activations must be transferred across GPUs, resulting in cumulative communication overhead that limits serving efficiency.

\subsection{Latent Parallelism}
\label{subsec:stp}
To reduce communication overhead in parallel serving of VDMs, we propose Latent Parallelism (LP). As shown in Figure \ref{parallel-serving}, LP distinguishes itself from conventional parallelism strategies by operating on the latent space of individual requests rather than partitioning the computational graph. The core idea is to decompose the global denoising problem by dynamically partitioning the latent tensor along different dimensions (temporal, height, width) at each timestep, with each partition (also called sub-latent) being processed in parallel on a different GPU. Figure \ref{framework} presents the workflow of LP, which consists of three stages: dynamic rotating partition, parallel denoising, and latent reconstruction.

\noindent
\textbf{Dynamic rotating partition.} To prevent visual discontinuity from incomplete spatio-temporal context, LP rotates the partitioning dimension at each timestep. For the $i$-th forward propagation of DiT blocks \textcolor{black}{(corresponding to diffusion timestep $t_i$, where $t_i=T+1-i$, counting from the initial noisy state),} the partitioning dimension $d_i$ is determined as follows:
\begin{equation}
d_i = \mathcal{M}[(i-1) \bmod 3 + 1],
\end{equation}
where $\mathcal{M}(\cdot)$ maps indices 1, 2, 3 to temporal, height, and width dimensions, respectively. \textcolor{black}{This rotating partition strategy ensures that each partitioned region captures the complete video context across all dimensions throughout the rotation cycles, maintaining global video consistency at the macroscopic level.}

\noindent
\textbf{Parallel denoising.} Given the current latent tensor $z_t$, forward propagation step $i$ with partitioning dimension $d_i$, LP partitions $z_t$ along $d_i$ into $K$ sub-regions $\{z_t^{(k)}\}_{k=1}^K$, where each sub-region undergoes parallel denoising computation on a different GPU:
\begin{equation}
\resizebox{0.48\textwidth}{!}{$\displaystyle
\tilde{f}_{\text{FP}}(z_t^{(k)}, t, c) = f(z_t^{(k)}, t, \emptyset) + w \cdot (f(z_t^{(k)}, t, c) - f(z_t^{(k)}, t, \emptyset)),
$}
\end{equation}
where $K$ corresponds to the number of available GPUs. Because the parallel denoising process operates on compact latent tensors instead of massive intermediate activations, LP significantly reduces inter-GPU communication overhead.
It is also worth noting that $f(\cdot)$ can be further parallelized using conventional strategies (NMP, TP, PP), allowing for hybrid inference acceleration.

\noindent
\textbf{Latent reconstruction.} 
After obtaining the noise predictions for all sub-regions, LP reconstructs them to form a complete latent for the current timestep. Specifically, the noise prediction at timestep $t$ is constructed via the following function $\mathcal{F}(\cdot)$:
\begin{equation}
\tilde{f}_{\text{FP}}(z_t, t, c) = \mathcal{F}\left(\{\tilde{f}_{\text{FP}}(z_t^{(k)}, t, c)\}_{k=1}^{K}, d_i\right).
\end{equation}
Then, the sampling scheduler $\mathcal{S}$ updates the latent tensor based on the reconstructed prediction:
\begin{equation}
z_{t-1} = \mathcal{S}\left(z_t, \tilde{f}_{\text{FP}}(z_t, t, c), t\right).
\end{equation}

These aforementioned stages are executed iteratively across all $T$ timesteps, ultimately generating the denoised latent $z_0$ from $z_T$ for decoding into video.
In the following sections, we will elaborate on the partitioning function $\mathcal{P}_{d_i}(\cdot)$ from $z_t \rightarrow \{z_t^{(k)}\}_{k=1}^K$ and the specific form of $\mathcal{F}(\cdot)$.

\subsection{Patch-Aligned Overlapping Partition}
\label{subsec:partition}
A critical challenge in partitioning the latent space is respecting the architectural constraints imposed by the internal structure of VDMs. Modern VDMs process latent tensors by dividing them into visual patches, which are spatio-temporal units that serve as the processing granularity of DiT blocks. Naive partitioning that cuts through patch boundaries may disrupt feature integrity within the denoising network and lead to visible artifacts. To address this challenge, we propose a patch-aligned overlapping partitioning mechanism that operates in patch space and incorporates overlapping regions to prevent boundary discontinuities.

\noindent
\textbf{Patch-aligned core region partition.} Let $(p_T, p_H, p_W)$ denote the patch sizes along the temporal, height, and width dimensions, respectively. For the current partitioning dimension $d_i$, we denote the patch size along this dimension as $p_{d_i}$. Given the latent tensor size $D_{d_i}$ along dimension $d_i$, the total number of patches along this dimension is $N_{d_i} = \left\lfloor \frac{D_{d_i}}{p_{d_i}} \right\rfloor$.
To distribute partitions across $K$ GPUs, we define $L = \left\lceil \frac{N_{d_i}}{K} \right\rceil$ as the number of patches along $d_i$ allocated to each GPU's core region (the non-overlapping portion of each partition). For the $k$-th partition where $k \in \{1, 2, \ldots, K\}$, the boundaries of its core region in patch space are as follows:
\begin{equation}
\alpha_k = (k-1) \cdot L, \quad \beta_k = \alpha_k + L
\label{partition-1}
\end{equation}
where $\alpha_k$ and $\beta_k$ represent the starting and ending patch indices of the core region, respectively.

\noindent
\textbf{Overlapping region extension.} As hard boundaries between adjacent partitions can lead to visible discontinuities in the generated video, we introduce controlled overlapping regions between partitions.
Specifically, the size of overlapping region is determined by the hyperparameter $r \in [0, K-1]$, which represents the ratio of overlapping patches to core patches in each partition. The number of overlapping patches is $O = \lfloor L \cdot r \rfloor$, and the extended boundaries $[\alpha_k', \beta_k')$ of the $k$-th partition are then:
\begin{equation}
\alpha_k' = \max(0, \alpha_k - O), \quad \beta_k' = \min(N_{d_i}, \beta_k + O),
\label{partition-2}
\end{equation}
where the $\max$ and $\min$ operations ensure that partition boundaries do not exceed the valid patch range $[0, N_{d_i})$.

\noindent
\textbf{Mapping back to latent space.}
To obtain the actual partition indices in the latent space, we map the patch space boundaries $[\alpha_k', \beta_k')$ back to the original latent space $[s_k, e_k)$ by scaling with the patch size:
\begin{equation}
s_k = \alpha_k' \cdot p_{d_i}, \quad e_k = \beta_k' \cdot p_{d_i}.
\label{partition-3}
\end{equation}
Therefore, the complete patch-aligned overlapping partition strategy is defined by the following mapping function:
\begin{equation}
\begin{aligned}
 \mathcal{P}_{d_i}(z_t, K, r) 
 = \Big\{ z_t^{(k)} = z_t[\mathcal{R}_k^{d_i}] \,\mid\, 
 \mathcal{R}_k^{d_i} = [s_k, e_k), \\
 k \in \{1, 2, \ldots, K\} \Big\},
\end{aligned}
\label{partition-4}
\end{equation}
where $z_t[\mathcal{R}_k^{d_i}]$ denotes extracting the sub-tensor with index range $\mathcal{R}_k^{d_i} = [s_k, e_k)$ along dimension $d_i$.

\subsection{Position-Aware Latent Reconstruction}
\label{subsec:fusion}
After parallel denoising across GPUs, we obtain $K$ local noise predictions $\{\tilde{f}_{\text{FP}}(z_t^{(k)}, t, c)\}_{k=1}^K$ that must be reconstructed into a complete global latent. Due to the overlapping regions introduced in section \ref{subsec:partition}, each overlapping position is predicted by multiple partitions, making direct concatenation infeasible. We address this through position-aware latent reconstruction, which seamlessly stitches these local predictions by assigning adaptive weights based on each position's spatio-temporal proximity to the core regions of the partitions.

\noindent
\textbf{Position-aware weight construction.} For the noise prediction $\tilde{f}_{\text{FP}}(z_t^{(k)}, t, c)$ of the $k$-th partition, we recall that its spatio-temporal extent along dimension $d_i$ is $\mathcal{R}_k^{d_i} = [s_k, e_k)$ with length $\ell_k = e_k - s_k$, and construct a weight mask $W^{(k)}$ with the same spatio-temporal dimensions as $\tilde{f}_{\text{FP}}(z_t^{(k)}, t, c)$.
Given the core region $[\alpha_k \cdot p_{d_i}, \beta_k \cdot p_{d_i})$ and the actual partition extent $[s_k, e_k)$ for GPU $k$, the distances from the partition boundaries to the core region boundaries are as follows:
\begin{equation}
\begin{aligned}
\Delta_k^{\text{start}} &= \alpha_k \cdot p_{d_i} - s_k = (\alpha_k - \alpha_k') \cdot p_{d_i}, \\
\Delta_k^{\text{end}} &= e_k - \beta_k \cdot p_{d_i} = (\beta_k' - \beta_k) \cdot p_{d_i},
\end{aligned}
\label{reconstruct-1}
\end{equation}
where $\Delta_k^{\text{start}}$ and $\Delta_k^{\text{end}}$ represent the lengths of the front and rear overlapping regions, respectively.
To achieve smooth reconstruction in overlapping regions, we adopt a local coordinate perspective and employ linear weights increase from 0 to 1 in the front overlapping region $[0, \Delta_k^{\text{start}})$, remain constant at 1 in the core region $[\Delta_k^{\text{start}}, \ell_k - \Delta_k^{\text{end}})$, and linearly decay from 1 to 0 in the rear overlapping region $[\ell_k - \Delta_k^{\text{end}}, \ell_k)$. Formally, for each local position $j \in [0, \ell_k)$ along dimension $d_i$, the weight is computed as:
\begin{equation}
W^{(k)}_j = \begin{cases}
\frac{j}{\Delta_k^{\text{start}}}, & 0 \leq j < \Delta_k^{\text{start}} ,\\[6pt]
1, & \Delta_k^{\text{start}} \leq j < \ell_k - \Delta_k^{\text{end}} ,\\[6pt]
\frac{\ell_k - j}{\Delta_k^{\text{end}}}, & \ell_k - \Delta_k^{\text{end}} \leq j < \ell_k,
\end{cases}
\label{reconstruct-2}
\end{equation}
which ensures that predictions from partitions farther from the core regions contribute less to the final reconstructed latent, and vice versa.

\noindent
\textbf{Coordinate localization.} To apply these partition-specific weights during reconstruction, we establish the mapping between global and local coordinate systems $\pi_k(\cdot)$ that converts a global coordinate $x$ along dimension $d_i$ to the corresponding local coordinate within the $k$-th partition:
\begin{equation}
\pi_k(x) = x - s_k, \quad \text{for } x \in [s_k, e_k).
\label{reconstruct-3}
\end{equation}
Additionally, we introduce the indicator function $\mathcal{I}_k(x)$ to denote whether a global position $x$ falls within the spatial extent of the $k$-th partition:
\begin{equation}
\mathcal{I}_k(x) =
\begin{cases}
1, & \text{if } s_k \leq x < e_k, \\
0, & \text{otherwise}.
\end{cases}
\label{reconstruct-4}
\end{equation}
Together, these functions enable us to bridge the gap between the global latent space and the local spaces of individual partitions.

\noindent
\textbf{Position-wise weighted averaging.} With the weight masks $\{W^{(k)}\}_{k=1}^K$ and coordinate mappings established, the reconstruction function $\mathcal{F}(\cdot)$ computes the all-global noise prediction through position-wise weighted averaging. For each global position $x$ along dimension $d_i$, we first accumulate the weighted predictions as follows:
\begin{equation}
A_{d_i}(x) = \sum_{k=1}^{K} \mathcal{I}_k(x) \cdot W^{(k)}_{\pi_k(x)} \cdot \left[\tilde{f}_{\text{FP}}(z_t^{(k)}, t, c)\right]_{\pi_k(x)},
\label{reconstruct-4-5}
\end{equation}
where $\left[\tilde{f}_{\text{FP}}(z_t^{(k)}, t, c)\right]_{\pi_k(x)}$ denotes the noise prediction tensor of the $k$-th partition at its local coordinate $\pi_k(x)$. Simultaneously, we compute the sum of weights from all contributing partitions:
\begin{equation}
Z_{d_i}(x) = \sum_{k=1}^{K} \mathcal{I}_k(x) \cdot W^{(k)}_{\pi_k(x)}.
\label{reconstruct-5}
\end{equation}
To this end, the reconstructed noise prediction at position $x$ is obtained by normalizing the weighted accumulation:
\begin{equation}
\left[\mathcal{F}\left(\{\tilde{f}_{\text{FP}}(z_t^{(k)}, t, c)\}_{k=1}^{K}, d_i\right)\right]_x = \frac{A_{d_i}(x)}{Z_{d_i}(x)},
\label{reconstruct-6}
\end{equation}
which is subsequently passed to the sampling scheduler, yielding the global latent.

\section{Theoretical Analysis for LP}
We provide theoretical analysis for LP from two fundamental perspectives: (1) Communication overhead: we quantitatively characterize the substantial reduction in communication achieved by LP compared to conventional parallelism strategies, demonstrating its superiority through dependence on compact latent tensors rather than high-dimensional activations for inter-GPU transfers. (2) Generation quality: we establish the theoretical foundation for how LP's dynamic rotating partition mechanism ensures that each sub-region eventually captures a complete spatio-temporal context by alternating decomposition across temporal, height, and width dimensions.
Due to page limitations, detailed proofs and derivations are provided in the supplementary material.

\section{Experiments}
\label{sec:experiments}

\subsection{Experimental Setup}
\label{subsec:setup}

\noindent
\textbf{Hardware platform.} Experiments are conducted on a physical cluster with 4 NVIDIA RTX A6000 GPUs interconnected via PCIe bus, equipped with 80 Intel(R) Xeon(R) Gold 5218R CPUs and 280GB of RAM.

\noindent
\textbf{Benchmarks.} We evaluate LP on three widely-used benchmarks: EvalCrafter \cite{liu2024evalcrafter}, T2V-CompBench \cite{sun2025t2v}, and VBench \cite{huang2024vbench}. Due to the substantial scale of datasets, we select representative subsets using a fixed random seed to ensure reproducibility.

\noindent
\textbf{Model configuration.} We employ WAN2.1-1.3B \cite{wan2025wan} as the video diffusion model, featuring a denoising network with 30 DiT Blocks, T5 text encoder, and pre-trained VAE decoder.

\noindent
\textbf{Baselines.} We compare LP against the following parallelism strategies for distributed VDM serving, implementing on PyTorch Distributed \cite{pytorch_distributed_docs} and xFusers \cite{fang2024xdit}.
\begin{itemize}
\item \textit{Naive Model Parallelism (NMP).} We evenly distribute DiT Blocks across GPUs in consecutive groups, with activations flowing sequentially through GPUs during inference.
\item \textit{Pipeline Parallelism (PP).} We apply the same model partition strategy as NMP, and exploit the conditional and unconditional forward passes of CFG to create a batch of size 2, enabling concurrent execution across stages and improving GPU utilization.
\item \textit{Hybrid Parallelism (HP).} We adopt the open-source inference framework developed by the Wan Team \cite{wan2-1-code}, which applys FSDP with xDiT that split individual operators at an intra-layer granularity to parallelize computation across multiple GPUs.
\end{itemize}

\noindent
\textbf{Metrics.} We employ two primary evaluation metrics: GPU communication overhead and video generation quality. For video generation quality assessment, we utilize VBench \cite{huang2024vbench} as our evaluation framework, focusing on five key sub-metrics: subject consistency (SC), background consistency (BC), temporal flickering (TF), motion smoothness (MS), and imaging quality (IQ). 
To contextualize the quality performance of parallelism strategies, we include two reference points: (1) Centralized, which generates videos on a single GPU without parallelization, and (2) VideoCrafter \cite{chen2023videocrafter1}, whose performance is reported in VBench.

\noindent
\textbf{Implementation details.} 
\textcolor{black}{Among the evaluated parallelism strategies, GPU 1 serves as the master orchestrator, executing its assigned DiT modules alongside encoder/decoder forward passes, sampler updates, and inter-GPU coordination at each denoising step.} Furthermore, all experiments adopt unified hyperparameter configurations: by default, 4 GPUs are used for parallelism, video frame rate is set to 16 FPS, resolution is 480p, and the denoising process adopts 60 iterations. To provide a comprehensive analysis, we evaluate both communication overhead and visual quality for videos of 49 and 81 frames, which correspond to 3-second and 5-second videos, respectively. For the LP method, we consider two cases with $r$ set to 0.5 and 1.0.

\begin{table}[t]
\centering
\caption{Communication overhead (MB) comparison across parallelism strategies.}
\vspace{-8pt}
\setlength{\tabcolsep}{2pt}
\begin{adjustbox}{width=\columnwidth}
\begin{tabular}{l|ccccc}
\hline
\multicolumn{1}{c|}{\multirow{2}{*}{\textbf{Method}}} & \multicolumn{5}{c}{\textbf{49 Frames, 480p, 3s}}                                         \\
\multicolumn{1}{c|}{}                                 & \textbf{GPU 1}   & \textbf{GPU 2}   & \textbf{GPU 3}   & \textbf{GPU 4}   & \textbf{Total}   \\ \hline
NMP                                        & 14563.59        & 14563.60        & 14563.60        & 14259.38        & 57950.17         \\
PP                                     & 14452.03        & 14439.38        & 14439.38        & 14259.38        & 57590.16         \\
HP                                       & 2084.44         & 891.21          & 891.21          & 891.21          & 4758.08          \\
LP ($r$=1.0)                          & 623.59          & 489.79          & 430.88          & 267.62          & 1811.88          \\
LP  ($r$=0.5)                          & \textbf{540.06} & \textbf{322.71} & \textbf{307.48} & \textbf{184.08} & \textbf{1354.34} \\ \hline
\multicolumn{1}{c|}{\multirow{2}{*}{\textbf{Method}}} & \multicolumn{5}{c}{\textbf{81 Frames, 480p, 5s}}                                         \\
\multicolumn{1}{c|}{}                                 & \textbf{GPU 1}   & \textbf{GPU 2}   & \textbf{GPU 3}   & \textbf{GPU 4}   & \textbf{Total}   \\ \hline
NMP                                        & 23338.59        & 23338.60        & 23338.60        & 23034.38        & 93050.17         \\
PP                                     & 23227.03        & 23214.38        & 23214.38        & 23034.38        & 92690.16         \\
HP                                       & 3367.18         & 1439.65         & 1439.65         & 1439.65         & 7686.12          \\
LP ($r$=1.0)                          & 993.28          & 770.10          & 703.07          & 446.37          & 2912.81          \\
LP ($r$=0.5)                          & \textbf{861.86} & \textbf{507.25} & \textbf{507.25} & \textbf{314.94} & \textbf{2191.29} \\ \hline
\end{tabular}
\end{adjustbox}
\label{comm-table}
\end{table}

\begin{figure}[t]  
    \centering
    \includegraphics[width=0.5\textwidth]{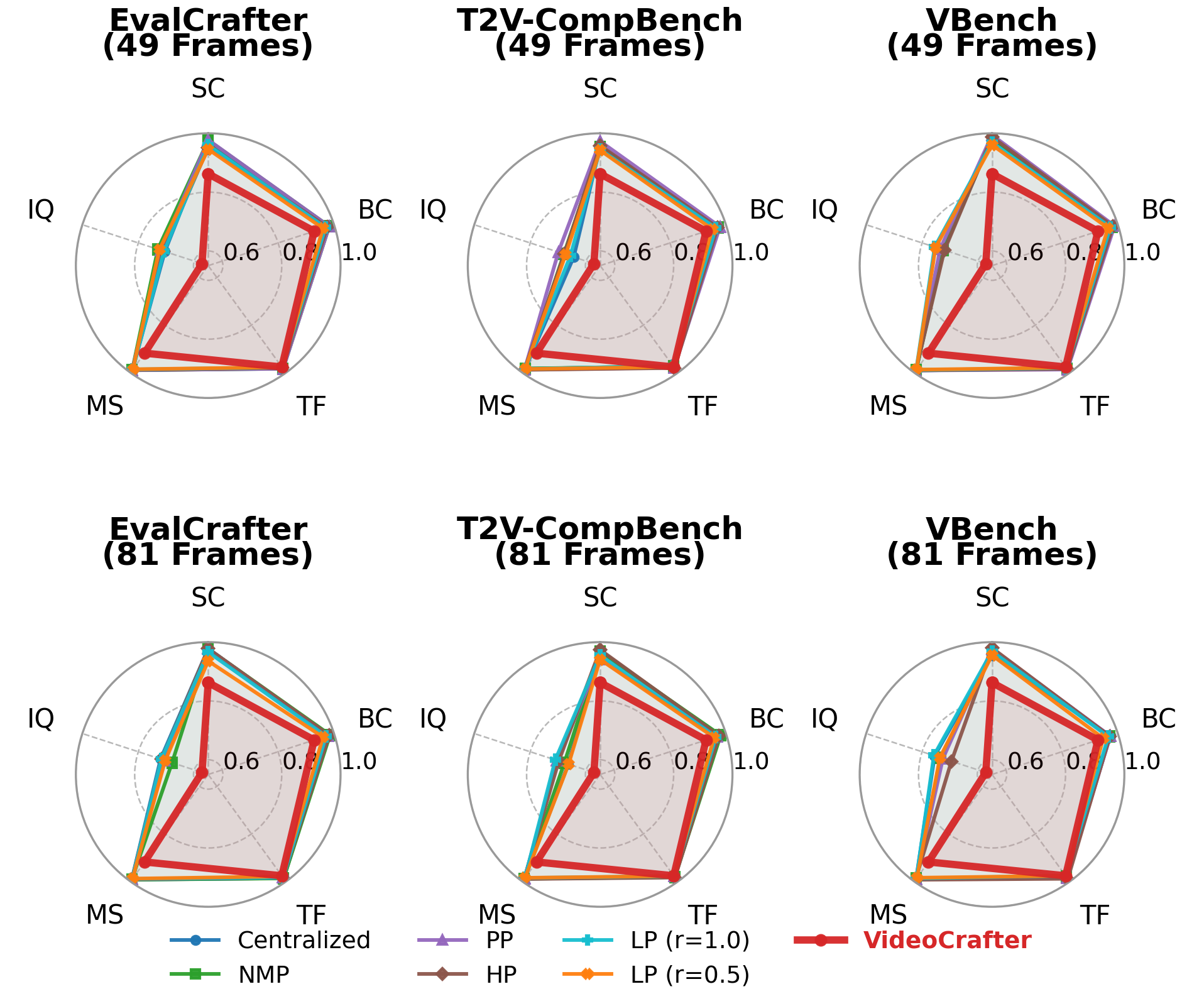}
    \vspace{-18pt}
    \caption{Comparison of video generation quality across three benchmarks.
}
    \label{radar}

\vspace{10pt}

    \centering
    \begin{minipage}{\linewidth}
        \centering
        \includegraphics[width=1.0\textwidth]{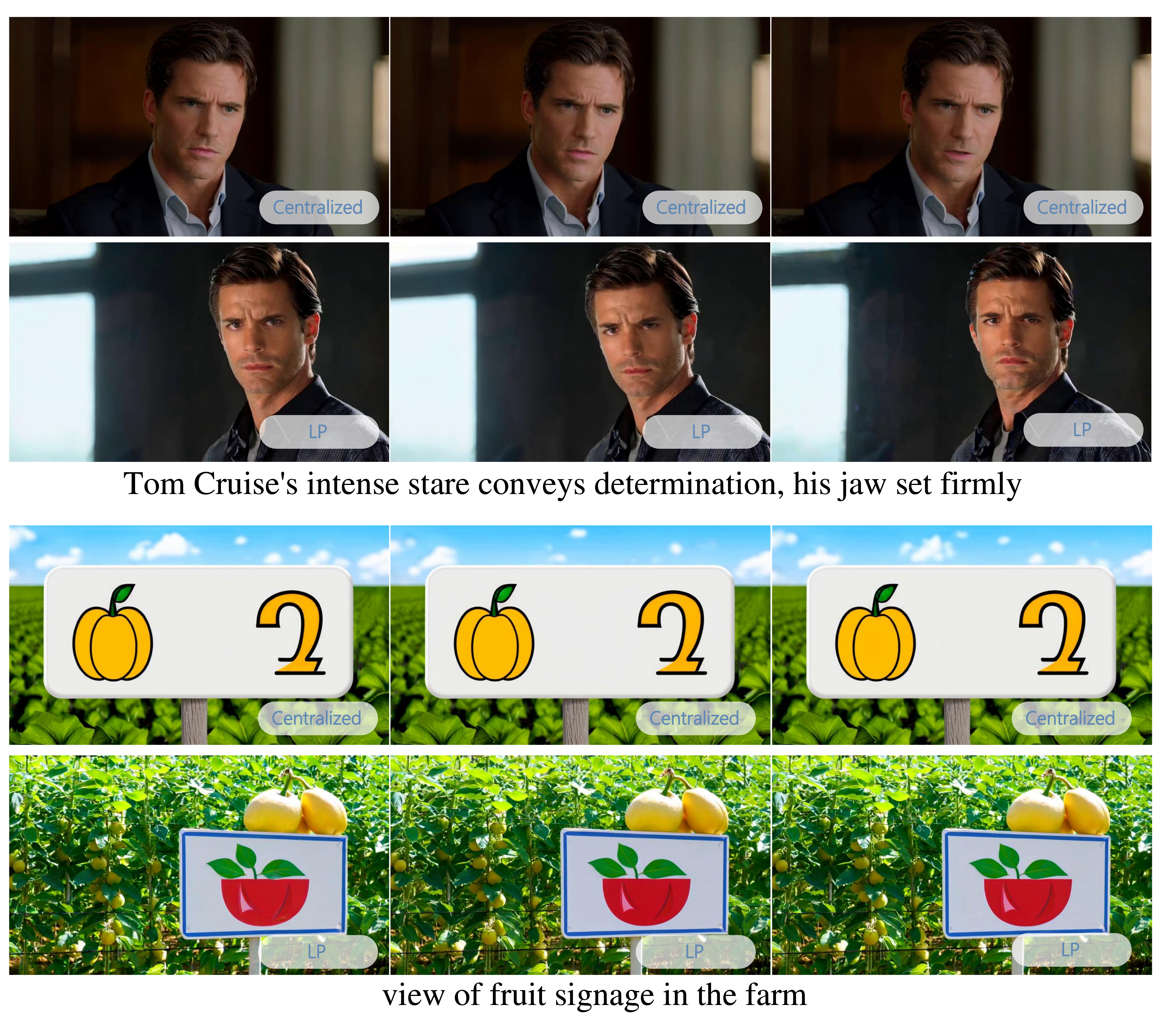}
        \vspace{-18pt}
        \caption{Visual comparison between Centralized and LP-generated videos. The left, middle, and right columns show the starting, middle, and ending frames of the videos, respectively.}
        \label{visualize}
    \end{minipage}\par\medskip
\end{figure}

\begin{figure*}[htbp]
    \centering 

    \begin{minipage}[t]{0.26\textwidth}
        \centering
        \includegraphics[width=\linewidth]{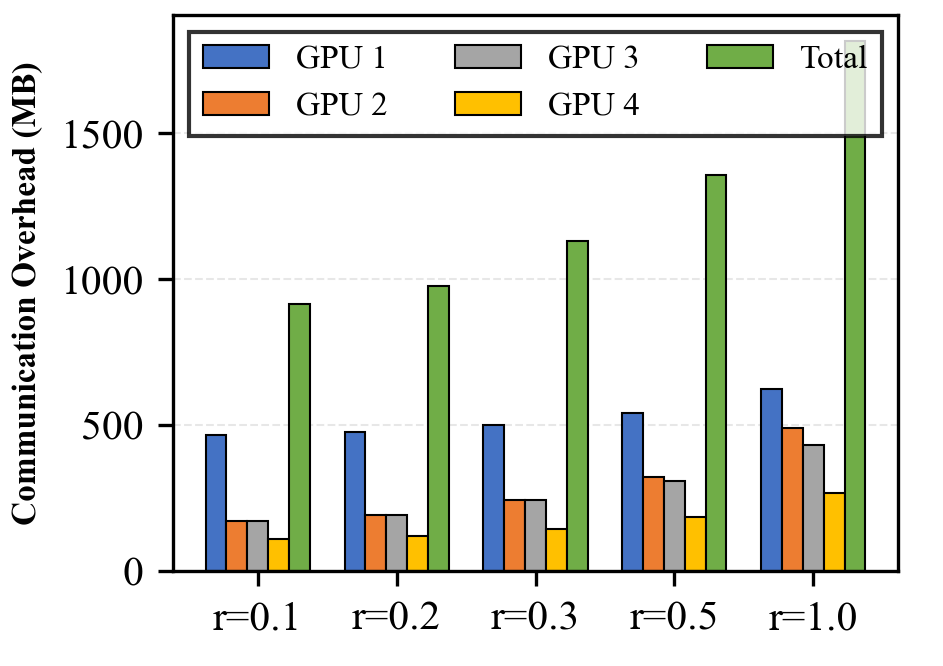}
        \vspace{-18pt}
        \captionof{figure}{Impact of overlap ratio on communication overhead.}
        \label{ablation-r-comm}
    \end{minipage}
    \hfill 
    %
    \begin{minipage}[t]{0.24\textwidth}
        \centering
        \includegraphics[width=\linewidth]{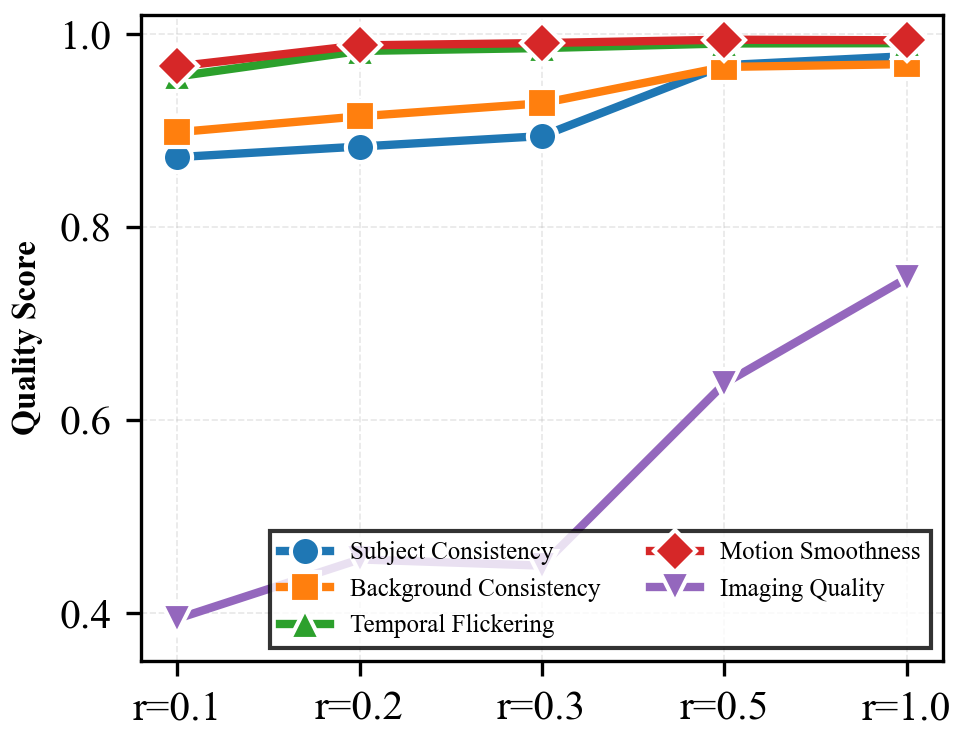}
        \vspace{-18pt}
        \captionof{figure}{Impact of overlap ratio on video generation quality.}
        \label{ablation-r-quality}
    \end{minipage}
    \hfill 
    \begin{minipage}[t]{0.23\textwidth}
        \centering
        \includegraphics[width=\textwidth]{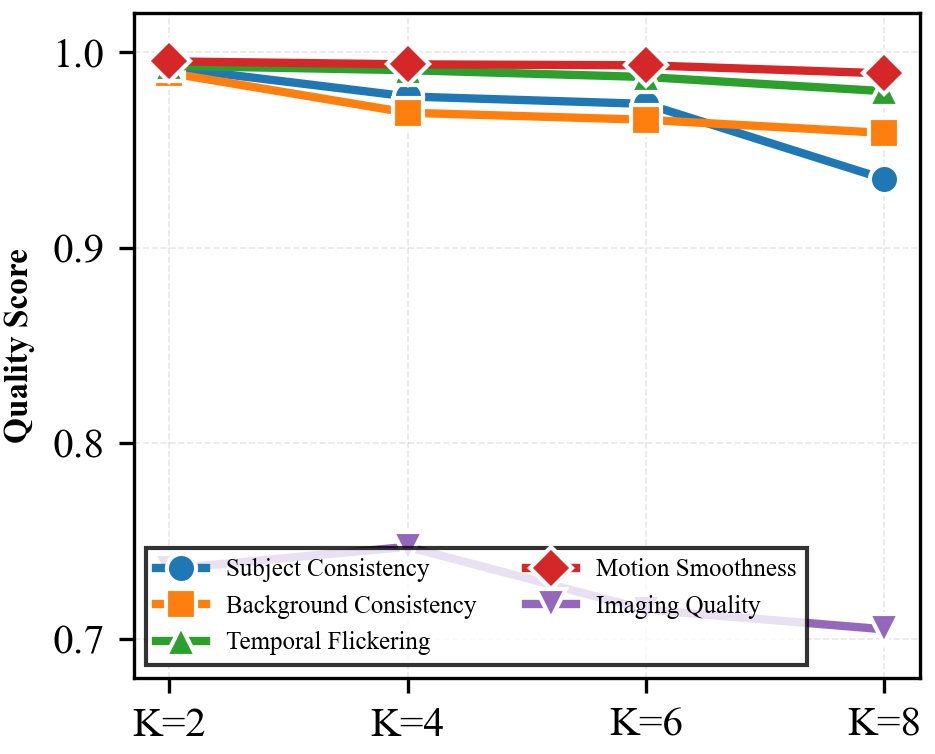}
        \vspace{-18pt}
        \caption{Impact of GPU numbers on video generation quality.}
        \label{diff-gpu}
    \end{minipage}
    \hfill 
    \begin{minipage}[t]{0.24\textwidth}
        \centering
        \includegraphics[width=\textwidth]{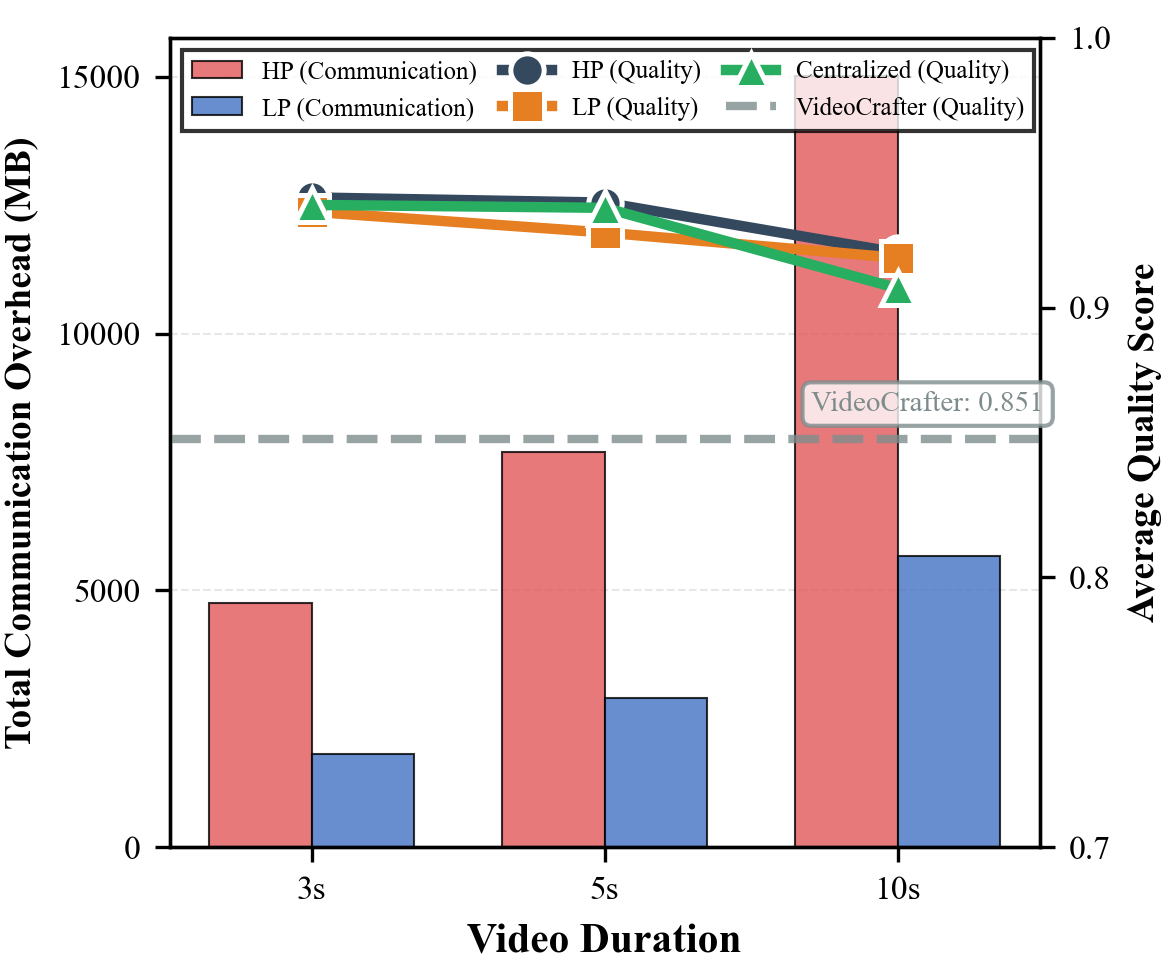}
        \vspace{-18pt}
        \caption{Impact of video duration on communication overhead and video generation quality.}
        \label{video-length}
    \end{minipage}

\end{figure*}

\begin{figure}[htbp]
    \centering 
    %
    \begin{minipage}[t]{0.43\textwidth}
        \centering
        \includegraphics[width=\linewidth]{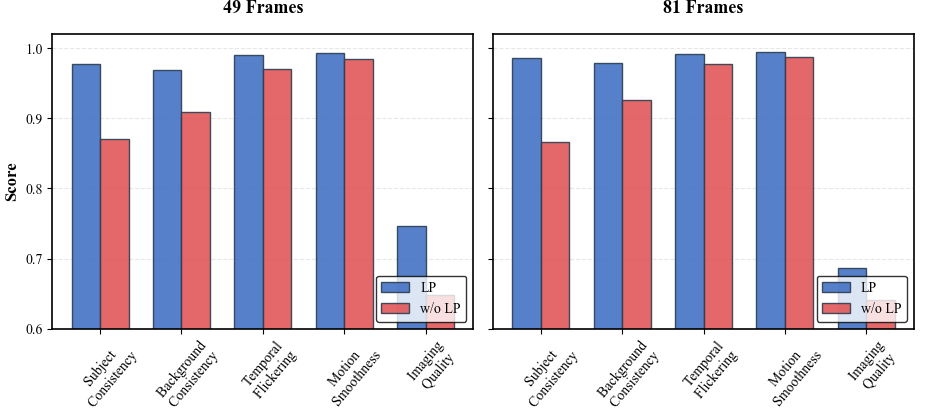}
        \vspace{-18pt}
        \captionof{figure}{Comparison of video generation quality: LP vs. baseline with only temporal partitioning (w/o LP).}
        \label{with-and-without-fp}
    \end{minipage}
\end{figure}

\subsection{Experimental Results}
\label{subsec:results}
\noindent
\textbf{Communication overhead.} Table \ref{comm-table} presents the communication overhead across different parallelism strategies for 49-frame (3s) and 81-frame (5s) video generation. 
Note that all benchmarks share the same communication overhead, which is independent of the specific prompt content. Traditional strategies (NMP, PP) exhibit severe communication burden, reaching 57GB and 92GB for 3s and 5s videos, respectively, as they repeatedly transfer high-dimensional activation tensors across GPUs. HP achieves reduced overhead compared with NMP and PP, but still requires over 4.6GB and 7.5GB of network traffic. In contrast, LP achieves dramatic reductions: with $r$=0.5, total communication overhead is merely 1.3GB (3s) and 2.1GB (5s), representing approximately 72\% reduction compared to HP and 97\% reduction over NMP and PP. This communication efficiency stems from LP's design, which transmits only compact latent tensors rather than high-dimensional activations.

\noindent
\textbf{Video generation quality.} Figure \ref{radar} presents quantitative quality evaluation results across three benchmarks using VBench metrics. The radar charts reveal two distinct performance tiers: LP, along with Centralized, NMP, PP, and HP, form an outer polygon with near-identical performance across all five dimensions (SC, BC, TF, MS, IQ), while VideoCrafter forms a contracted inner polygon with notably lower scores, particularly in IQ, SC, and MS. This demonstrates that LP achieves substantial communication savings without compromising video generation quality. Notably, LP consistently outperforms VideoCrafter across both 49- and 81-frame settings and all three benchmarks, confirming its robustness and generalization capability independent of video length or evaluation scenario.

\noindent
\textbf{Qualitative visualization.} Figure \ref{visualize} shows a qualitative comparison between Centralized and LP-generated videos. The selected frames demonstrate that LP maintains visual fidelity comparable to centralized generation across diverse scenarios. For instance, in the "Tom Cruise" prompt, both methods capture facial details and emotional expressions with equivalent quality. Similarly, for the "fruit signage" prompt, LP preserves texture clarity and color accuracy without introducing boundary artifacts or spatial discontinuities, indicating that our patch-aligned overlapping partition and position-aware latent reconstruction mechanisms effectively stitch the parallel-denoised sub-latents seamlessly. These visual results qualitatively corroborate our quantitative findings, suggesting that LP maintains video generation quality.

\subsection{Ablation Study}
\label{subsec:ablation}
We conduct a series of ablation studies to analyze the impact of key design choices and hyperparameters in LP. By default, we use the experimental configuration with $r=1.0$, 49-frame, and the first prompt sampled from EvalCrafter.

\noindent
\textbf{Impact of overlap ratio.} Figures \ref{ablation-r-comm} and \ref{ablation-r-quality} reveal that \textcolor{black}{while communication overhead of LP doubles from approximately 900MB to 1800MB as $r$ increases from 0.1 to 1.0, it still maintains a substantial reduction over HP (4758MB).} Meanwhile, video generation quality improves with increasing $r$ but plateaus beyond $r=0.5$, suggesting this value offers an optimal trade-off.

\noindent
\textbf{Impact of GPU Numbers.} Figure \ref{diff-gpu} demonstrates LP's scalability across 2-8 GPUs. As shown, LP exhibits consistently high video quality, superior to VideoCrafter across all configurations with minor variation, confirming its robustness to different cluster configurations.

\noindent
\textbf{Impact of video duration.} Figure \ref{video-length} illustrates LP's superior scalability for longer videos. While HP's overhead escalates from about 5GB to 15GB for 3s to 10s videos, LP's overhead increases by only 4GB. Notably, LP maintains quality metrics on par with HP and Centralized across all durations while significantly outperforming VideoCrafter, confirming its capability to handle extended video generation without compromising generation quality.

\noindent
\textbf{Ablation on partitioning strategy.} Figure \ref{with-and-without-fp} compares LP against a baseline using temporal-only partitioning (w/o LP). As shown, LP consistently outperforms the baseline across all metrics for both 49 and 81 frames, confirming that its dynamic rotating partition, patch-aligned overlapping, and position-aware latent reconstruction mechanisms work synergistically to preserve generation quality.

\noindent
\textbf{Analysis of end-to-end latency.} As shown in Table \ref{e2d-latency}, LP consistently delivers lower inference latency compared to NMP, even with a large overlap ratio ($r=1.0$). These results confirm that latent-space partitioning paradigm effectively addresses the communication bottleneck, translating reduced network traffic into tangible speedups for practical VDM serving.

\begin{table}[t]
\caption{End-to-end latency across two parallelism strategies.}
\vspace{-8pt}
\centering
\begin{tabular}{c|ccc}
\hline
        & NMP    & LP ($r$=1.0) & LP ($r$=0.5) \\ \hline
Latency & 239.33s & 220.69s     & \textbf{195.27s}     \\ \hline
\end{tabular}
\label{e2d-latency}
\end{table}

\section{Conclusion}
\label{sec:conclusion}
We propose LP, the first parallelism strategy tailored for VDM serving. Exploiting the local spatio-temporal dependency in the diffusion denoising process and the compactness of latent space, LP employs dynamic rotating partitions across temporal, height, and width dimensions at different timesteps, decomposing the global denoising task into parallelizable sub-problems with only lightweight latent tensors being communicated between GPUs. To preserve generation quality, we design patch-aligned overlapping partition that aligns with visual patch boundaries, and position-aware latent reconstruction that eliminates boundary artifacts through adaptive weighting. Experiments demonstrate that LP substantially reduces communication overhead compared to conventional parallelism strategies, while maintaining comparable video generation quality.



{
    \small
    \bibliographystyle{ieeenat_fullname}
    
}

\clearpage
\setcounter{page}{1}
\maketitlesupplementary

\section{Theoretical Analysis of LP's Communication Overhead}
In this section, we provide a rigorous formalization of the communication overhead in LP. We proceed under the assumption of an optimal LP implementation; this assumption applies to all subsequent analyses in this paper. Without loss of generality, we derive the total communication overhead of LP as $C_{LP}$, and perform quantitative comparisons against NMP ($C_{NMP}$) and PP ($C_{PP}$). Our analysis focuses on the primary communication components, which are the activation tensors transferred in NMP and PP as well as the video latents in LP, while excluding other overhead factors.

\subsection{Additional Notations}
We first introduce additional notations beyond those established in the main paper. Let $S_H$ denote the size of the high-dimensional intermediate activation tensor, and $S_z$ be the size of the complete global latent tensor. Under the LP strategy, $S_{sub}^{(k)}$ denotes the size of the sub-latent tensor assigned to the $k$-th GPU. To simplify the expression, we define the total size of all sub-latent tensors $S_{ext}$ as follows:
\begin{equation}
    S_{ext} = \sum_{k=1}^K S_{sub}^{(k)}.
\end{equation}
We further define the latent space expansion factor $\gamma(r, K)$ as follows:
\begin{equation}
    \gamma(r, K) = \frac{S_{ext}}{S_z},
    \label{pengzhang}
\end{equation}
which satisfies $\gamma \ge 1$ and increases monotonically with both $r$ and $K$. Finally, we use $C_{\text{cond\_pass}}$ and $C_{\text{uncond\_pass}}$ to denote the communication overhead for single conditional and unconditional forward pass in CFG, respectively, while $C_{NMP}, C_{PP}, C_{LP}$ represent the total communication overhead across $T$ denoising steps for each parallelism strategy.

\subsection{Communication Overhead of Baseline Parallelism Strategies}
We first analyze the communication overhead of baseline parallelism strategies.

\noindent
\textbf{Communication overhead of NMP.} In NMP, the DiT blocks of model $f(\cdot)$ are evenly distributed across $K$ GPUs. For a single forward pass of $f(\cdot)$, the computational flow must traverse $K$ GPUs sequentially, incurring communication at $K-1$ GPU boundaries. At each boundary, the high-dimensional activation tensor must be transferred. Therefore, the overhead for a single pass is:
\begin{equation}
C_{\text{pass, NMP}} = (K-1) \cdot S_H.
\end{equation}
At any timestep $t$, the communication overhead of NMP is determined by both conditional and unconditional executions required by CFG:
\begin{equation}
\begin{aligned}
&\;\;\;\;C_{\text{step, NMP}} \\&= C_{\text{cond\_pass}} + C_{\text{uncond\_pass}} \\
&= [(K-1) \cdot S_H] + [(K-1) \cdot S_H] \\&= 2(K-1) \cdot S_H.
\end{aligned}
\label{comm-nmp}
\end{equation}
Since the entire denoising process comprises $T$ steps, the total communication overhead of NMP is:
\begin{equation}
C_{NMP} = T \cdot C_{\text{step, NMP}} = 2T(K-1) \cdot S_H.
\label{nmp-comm-total}
\end{equation}

\noindent
\textbf{Communication overhead of PP.} 
In PP, the DiT block distribution across GPUs is identical to that in NMP. As described in the experimental section, PP exploits the two passes of CFG as micro-batches of size 2 to construct a pipeline. However, PP does not reduce the volume of tensors that must be transferred.
To complete $T$ denoising steps, the activation tensor for the total $2T$ forward passes must still traverse $K-1$ GPU boundaries:
\begin{equation}
C_{PP} = 2T(K-1) \cdot S_H
\end{equation}
Therefore, in terms of total communication volume approximation, we have $C_{PP} = C_{NMP}$. We will use this as the baseline for comparison with LP's communication overhead.

\subsection{Communication Overhead of LP}
\label{comm-cost-lp}
In LP, the inference paradigm undergoes a fundamental shift. Rather than partitioning the model $f(\cdot)$ across GPUs, we partition the relatively compact latent tensor across GPUs, and the communication occurs during the following two stages: dynamic rotating partition and latent reconstruction.

\noindent
\textbf{Communication overhead of dynamic rotating partition.}
In this stage, the master GPU partitions the global latent $z_t$ into $K$ overlapping sub-latent tensors $\{z_t^{(k)}\}_{k=1}^K$. It then scatters $K-1$ sub-latent tensors to the other GPUs:
\begin{equation}
C_{\text{Scatter}} = \sum_{k=2}^K S_{sub}^{(k)}
\end{equation}

\noindent
\textbf{Communication overhead of latent reconstruction.}
After parallel denoising, each worker GPU sends its computed local noise prediction $\tilde{f}_{FP}(z_t^{(k)})$ (with size equal to $z_t^{(k)}$) back to the master GPU, which gathers these predictions for position-aware latent reconstruction:
\begin{equation}
C_{\text{Gather}} = \sum_{k=2}^K S_{sub}^{(k)}
\end{equation}
Due to the CFG framework, LP requires parallel execution of two forward passes at each diffusion step. This means both scatter and gather operations must be performed twice. Therefore, the per-step communication overhead of LP is as follows:
\begin{equation}
\begin{aligned}
& \;\;\;\; C_{\text{step, LP}} \\
&= (C_{\text{Scatter}} + C_{\text{Gather}}) \times 2 \\
&= \left[ 2 \cdot \sum_{k=2}^K S_{sub}^{(k)} \right] \times 2 \\ & = 4 \cdot \sum_{k=2}^K S_{sub}^{(k)}.
\end{aligned}
\label{step-lp}
\end{equation}
Across the entire $T$-step denoising process, the total communication overhead of LP is:
\begin{equation}
C_{LP} = T \cdot C_{\text{step, LP}} = 4T \cdot \sum_{k=2}^K S_{sub}^{(k)}.
\end{equation}

\noindent
\textbf{Overhead approximation.}
To relate $C_{LP}$ to $S_z$, we approximate that the sub-latent tensors are roughly balanced across $K$ GPUs, which means:
\begin{equation}
    S_{sub}^{(k)} \approx \frac{S_{ext}}{K}.
\end{equation}
Then, 
\begin{equation}
    \sum_{k=2}^K S_{sub}^{(k)} \approx (K-1) \cdot \frac{S_{ext}}{K} = \frac{K-1}{K} S_{ext}.
\end{equation}
Substituting $S_{ext} = \gamma(r, K) \cdot S_z$ according to Eq. (\ref{pengzhang}), we have:
\begin{equation}
C_{LP} \approx 4T \frac{K-1}{K} \cdot \gamma(r, K) \cdot S_z.
\end{equation}

\subsection{Comparative Analysis}

We can now derive the communication overhead ratio of LP relative to NMP and PP. For NMP, we have:
\begin{equation}
\begin{aligned}
& \;\;\;\; R_{LP/NMP} \\
& = \frac{C_{LP}}{C_{NMP}} \\
& \approx \frac{4T \frac{K-1}{K} \gamma(r, K) \cdot S_z}{2T(K-1) \cdot S_H} \\
& \approx \frac{4 \cdot \gamma(r, K) \cdot S_z}{2 \cdot K \cdot S_H} \\
& \approx \frac{2 \cdot \gamma(r, K)}{K} \cdot \left( \frac{S_z}{S_H} \right).
\end{aligned}
\end{equation}
For PP, the ratio is identical as $C_{NMP} = C_{PP}$:
\begin{equation}
\begin{aligned}
& \;\;\;\;
R_{LP/PP} \\&= \frac{C_{LP}}{C_{PP}} \\&=R_{LP/NMP} \\&\approx \frac{2 \cdot \gamma(r, K)}{K} \cdot \left( \frac{S_z}{S_H} \right).
\end{aligned}
\end{equation}
The derived ratio $R \approx \frac{2 \cdot \gamma(r, K)}{K} \cdot \left( \frac{S_z}{S_H} \right)$ theoretically reveals the source of LP's communication superiority:

\begin{itemize}
    \item \textbf{Critical ratio} $\frac{S_z}{S_H}$. The core of LP's communication efficiency lies in $S_z \ll S_H$, as the video latent is significantly more compact than high-dimensional activation tensors within DiT blocks. For models like WAN2.1, empirical values of this ratio are as low as $\sim$5\%.
    \item \textbf{Constant factor} $\frac{2 \cdot \gamma(r, K)}{K}$. Under all circumstances, $\frac{\gamma(r, K)}{K}$ will be a factor bounded by 1, implying that $\frac{2 \cdot \gamma(r, K)}{K}$ is always at most 2.
\end{itemize}
Ultimately, the communication overhead ratio $R$ is dominated by the extremely small term $\frac{S_z}{S_H}$. This perfectly explains why LP achieves up to 97\% reduction in communication overhead, as demonstrated in the main paper.

In summary, while NMP and PP incur communication overhead bound by the massive $S_H$, LP's overhead is bound by the compact $S_z$. Our theoretical analysis demonstrates that this advantage $R \propto \frac{S_z}{S_H}$ stems from the fundamental paradigm shift in VDM parallelism strategies: partitioning the latent space rather than partitioning the model.

\section{Theoretical Analysis of LP's Video Generation Quality}
In this section, we provide a theoretical analysis establishing how LP preserves video generation quality. We adopt a receptive field perspective to rigorously demonstrate that LP maintains global consistency by ensuring complete spatio-temporal context coverage throughout the denoising process.

\subsection{Preliminaries and Key Definitions}
We introduce three foundational concepts to formalize the quality preservation guarantees of LP.

\begin{definition}[Receptive Field]
The receptive field of a position $p$ in the latent space after $i$ denoising steps, denoted as $\mathcal{R}(p, i)$, characterizes the set of all spatio-temporal positions whose information has been exchanged with position $p$ through the diffusion process.
\end{definition}

\begin{definition}[Completeness]
A parallel inference method is \emph{complete} if it enables the receptive field of any position to eventually cover the entire spatio-temporal latent space $Z$ during the diffusion process, i.e., $\exists N \in \mathbb{N}$ such that $\mathcal{R}(p, N) = Z$ for all $p \in Z$. This property ensures global consistency in the generated video.
\end{definition}
Obviously, centralized inference (processing the entire latent space on a single GPU) is trivially complete, as all positions directly interact within each denoising step. In contrast, naive parallelization strategies that partition solely along a single dimension (e.g., w/o LP) are incomplete, as positions in different partitions cannot exchange information across the partitioned dimension.

\begin{definition}[$N$-Completeness]
A parallel inference method is \emph{$N$-complete} if the receptive field of any position can cover the entire global latent space within $N$ consecutive denoising steps, i.e., $\mathcal{R}(p, N) = Z$ for all $p \in Z$. This quantifies the efficiency of global information propagation.
\end{definition}


\subsection{Main Theorem and Proof}

We now establish the key theoretical result that guarantees LP's quality preservation.

\begin{theorem}[2-Completeness of LP]
\label{thm:lp-completeness}
Latent Parallelism is a 2-complete parallelism strategy. Formally, for any position $p$ in the latent space $Z$, we have $\mathcal{R}(p, 2) = Z$.
\end{theorem}

\begin{proof}
Let $Z$ denote the video latent space that spans temporal, height, and width dimensions, which can be expressed as $Z \sim D_T \times D_H \times D_W$. Any position $p \in Z$ is uniquely identified by its coordinates $p = (t, h, w)$ where $t \in [0, D_T)$, $h \in [0, D_H)$, and $w \in [0, D_W)$.

\paragraph{Step 1: revisit the partition strategy.}
As described in the main paper, the partition dimension $d_i$ at the forward propagation step $i$ is determined by:
\begin{equation}
d_i = \mathcal{M}[(i-1) \bmod 3 + 1],
\end{equation}
where $\mathcal{M}(\cdot)$ maps indices 1, 2, 3 to temporal, height, and width dimensions, respectively. This rotating partition strategy ensures that any two consecutive steps $i$ and $i+1$ partition along different dimensions.

\paragraph{Step 2: general case analysis.}
Without loss of generality, we analyze the case where step $i$ partitions along the temporal dimension and step $i+1$ partitions along the height dimension. The analysis for any other pair of consecutive different dimensions (e.g., temporal/width or height/width) follows by symmetry.

\paragraph{Step 3: receptive field after first step.}
At step $i$, the latent space $Z$ is partitioned along the temporal dimension into $K$ sub-latents $\{z_{t_i}^{(k)}\}_{k=1}^K$. Consider a position $p = (t, h, w)$. Due to the dynamic rotating partition strategy, position $p$ resides in one or more sub-latents $z_{t_i}^{(k)}$ that cover the temporal partition $\mathcal{P}_T(t)$ containing coordinate $t$, while spanning the complete height and width dimensions, which means:
\begin{equation}
z_{t_i}^{(k)} \sim \mathcal{P}_T(t) \times [0, D_H) \times [0, D_W).
\end{equation}
During parallel denoising, the DiT blocks perform self-attention computation within each sub-latent $z_{t_i}^{(k)}$. Since the attention mechanism operates globally across all positions within the sub-latent, information at position $p$ is fused with all positions sharing the same temporal partition. Consequently, the receptive field after one step can be expressed as follows:
\begin{equation}
\begin{aligned}
\mathcal{R}(p, 1) = \{(t', h', w') \mid t' &\in \mathcal{P}_T(t), \\
&h' \in [0, D_H), w' \in [0, D_W)\}.
\end{aligned}
\end{equation}
This indicates that information from position $p$ has achieved global propagation across the height and width dimensions, while maintaining locality in the temporal dimension.

\paragraph{Step 4: receptive field after second step.}
After latent reconstruction at step $i$, the updated latent space is repartitioned along the height dimension at step $i+1$. Consider any position $p_1 = (t_1, h_1, w_1) \in \mathcal{R}(p, 1)$, which already carries information from the original position $p$. Position $p_1$ is assigned to a new sub-latent $z_{t_{i+1}}^{(k')}$ covering the height partition $\mathcal{P}_H(h_1)$ containing $h_1$, while spanning the complete temporal and width dimensions:
\begin{equation}
z_{t_{i+1}}^{(k')} \sim [0, D_T) \times \mathcal{P}_H(h_1) \times [0, D_W).
\end{equation}
Through attention computation at step $i+1$, information from $p_1$ (which contains information from $p$) propagates to all positions within $z_{t_{i+1}}^{(k')}$. The receptive field of $p_1$ after this second step becomes:
\begin{equation}
\begin{aligned}
\mathcal{R}(p_1, 1) = \{(t'', h'', w'') \mid t'' &\in [0, D_T), \\
&h'' \in \mathcal{P}_H(h_1), w'' \in [0, D_W)\}.
\end{aligned}
\end{equation}
To this end, the receptive field of the original position $p$ after two denoising steps is obtained by taking the union over all intermediate positions:
\begin{equation}
\begin{aligned}
& \;\;\;\; \mathcal{R}(p, 2) \\ & = \bigcup_{p_1 \in \mathcal{R}(p, 1)} \mathcal{R}(p_1, 1) \\
&= \bigcup_{(t_1, h_1, w_1) \in \mathcal{R}(p, 1)} \{(t'', h'', w'') \mid \\
&\quad t'' \in [0, D_T), h'' \in \mathcal{P}_H(h_1), w'' \in [0, D_W)\}.
\end{aligned}
\end{equation}

\paragraph{Step 5: complete coverage analysis.}
We now analyze the coverage across each dimension:
\begin{itemize}
    \item Temporal dimension: Each receptive field $\mathcal{R}(p_1, 1)$ spans the complete range $t'' \in [0, D_T)$, thus $\mathcal{R}(p, 2)$ fully covers the temporal dimension.
    
    \item Width dimension: Similarly, each $\mathcal{R}(p_1, 1)$ spans $w'' \in [0, D_W)$, ensuring complete width coverage.
    
    \item Height dimension: Since $\mathcal{R}(p, 1)$ already covers the complete height range $h_1 \in [0, D_H)$ from the first step, and the union is taken over all such positions, we have:
    \begin{equation}
    \bigcup_{h_1 \in [0, D_H)} \mathcal{P}_H(h_1) = [0, D_H).
    \end{equation}
    This follows from the fact that the height partitions $\{\mathcal{P}_H(h)\}$ collectively cover the entire height dimension.
\end{itemize}
Therefore, we can conclude that:
\begin{equation}
\begin{aligned}
& \;\;\;\; \mathcal{R}(p, 2) \\&= \{(t'', h'', w'') \mid t'' \in [0, D_T), 
\\
&\;\;\;\; h'' \in [0, D_H), w'' \in [0, D_W)\} \\&= Z.
\end{aligned}
\end{equation}
This demonstrates that the receptive field of any position $p$ covers the entire latent space $Z$ after two consecutive denoising steps with different partition dimensions. 
Since LP's dynamic rotating partition strategy ensures that consecutive partitions always operate along different dimensions, any position achieves global coverage within two steps.
Therefore, LP is a 2-complete parallelism strategy.
\end{proof}

\subsection{Implications for Video Generation Quality}
The 2-completeness property established in Theorem~\ref{thm:lp-completeness} has significant implications for preserving video generation quality. On the one hand, it ensures that every position in the latent space accesses complete spatio-temporal context within just two denoising steps. Given state-of-the-art methods typically require $T \geq 20$ denoising steps, each position can undergo approximately $\lfloor T/2 \rfloor$ complete information exchange cycles throughout the diffusion process. This repeated global interaction prevents the accumulation of local inconsistencies that could manifest as temporal flickering, spatial discontinuities, or motion artifacts. On the other hand, 2-completeness establishes that LP provides equivalent information flow to centralized inference. This theoretical equivalence explains why LP maintains video quality comparable to centralized inference across all benchmarks in our experiments.

\begin{table*}[t]
\caption{Numerical comparison of video generation quality across three datasets. Note that the results for VideoCrafter are reported from the VBench benchmark \cite{huang2024vbench}.}
\setlength{\tabcolsep}{2.5pt} 
\renewcommand{\arraystretch}{0.9}
\centering
\begin{tabular}{l|c|lccccc|c}
\hline
\multicolumn{1}{c|}{\multirow{24}{*}{\textbf{49 Frames}}} & \multirow{8}{*}{\textbf{EvalCrafter}}   & \multicolumn{1}{c}{\textbf{Method}} & \textbf{\begin{tabular}[c]{@{}c@{}}Subject\\ Consistency\end{tabular}} & \textbf{\begin{tabular}[c]{@{}c@{}}Background\\ Consistency\end{tabular}} & \textbf{\begin{tabular}[c]{@{}c@{}}Temporal\\ Flickering\end{tabular}} & \textbf{\begin{tabular}[c]{@{}c@{}}Motion\\ Smoothness\end{tabular}} & \textbf{\begin{tabular}[c]{@{}c@{}}Imaging\\ Quality\end{tabular}} & \textbf{Average} \\ \cline{3-9} 
\multicolumn{1}{c|}{}                                     &                                         & Centralized                         & 0.9719                                                                 & 0.9744                                                                    & 0.9814                                                                 & 0.9885                                                               & 0.7051                                                             & 0.9243           \\
\multicolumn{1}{c|}{}                                     &                                         & NMP                                 & 0.9763                                                                 & 0.9825                                                                    & 0.9830                                                                 & 0.9897                                                               & 0.7301                                                             & 0.9323           \\
\multicolumn{1}{c|}{}                                     &                                         & PP                                  & 0.9775                                                                 & 0.9825                                                                    & 0.9833                                                                 & 0.9897                                                               & 0.7186                                                             & 0.9303           \\
\multicolumn{1}{c|}{}                                     &                                         & HP                                  & 0.9522                                                                 & 0.9757                                                                    & 0.9795                                                                 & 0.9886                                                               & 0.7108                                                             & 0.9214           \\
\multicolumn{1}{c|}{}                                     &                                         & LP ($r$=1.0)                          & 0.9598                                                                 & 0.9694                                                                    & 0.9777                                                                 & 0.9865                                                               & 0.7079                                                             & 0.9203           \\
\multicolumn{1}{c|}{}                                     &                                         & LP ($r$=0.5)                          & 0.9457                                                                 & 0.9605                                                                    & 0.9770                                                                 & 0.9855                                                               & 0.7248                                                             & 0.9187           \\
\multicolumn{1}{c|}{}                                     &                                         & VideoCrafter                       & 0.8624                                                                 & 0.9288                                                                    & 0.9760                                                                 & 0.9179                                                               & 0.5722                                                             & 0.8515           \\ \cline{2-9} 
\multicolumn{1}{c|}{}                                     & \multirow{8}{*}{\textbf{T2V-CompBench}} & \multicolumn{1}{c}{\textbf{Method}} & \textbf{\begin{tabular}[c]{@{}c@{}}Subject\\ Consistency\end{tabular}} & \textbf{\begin{tabular}[c]{@{}c@{}}Background\\ Consistency\end{tabular}} & \textbf{\begin{tabular}[c]{@{}c@{}}Temporal\\ Flickering\end{tabular}} & \textbf{\begin{tabular}[c]{@{}c@{}}Motion\\ Smoothness\end{tabular}} & \textbf{\begin{tabular}[c]{@{}c@{}}Imaging\\ Quality\end{tabular}} & \textbf{Average} \\ \cline{3-9} 
\multicolumn{1}{c|}{}                                     &                                         & Centralized                         & 0.9560                                                                 & 0.9686                                                                    & 0.9735                                                                 & 0.9834                                                               & 0.6444                                                             & 0.9052           \\
\multicolumn{1}{c|}{}                                     &                                         & NMP                                 & 0.9526                                                                 & 0.9735                                                                    & 0.9739                                                                 & 0.9825                                                               & 0.6793                                                             & 0.9124           \\
\multicolumn{1}{c|}{}                                     &                                         & PP                                  & 0.9731                                                                 & 0.9784                                                                    & 0.9787                                                                 & 0.9869                                                               & 0.6993                                                             & 0.9232           \\
\multicolumn{1}{c|}{}                                     &                                         & HP                                  & 0.9583                                                                 & 0.9698                                                                    & 0.9802                                                                 & 0.9879                                                               & 0.6780                                                             & 0.9149           \\
\multicolumn{1}{c|}{}                                     &                                         & LP ($r$=1.0)                          & 0.9449                                                                 & 0.9639                                                                    & 0.9718                                                                 & 0.9840                                                               & 0.6564                                                             & 0.9042           \\
\multicolumn{1}{c|}{}                                     &                                         & LP ($r$=0.5)                          & 0.9421                                                                 & 0.9511                                                                    & 0.9758                                                                 & 0.9844                                                               & 0.6731                                                             & 0.9053           \\
\multicolumn{1}{c|}{}                                     &                                         & VideoCrafter                        & 0.8624                                                                 & 0.9288                                                                    & 0.9760                                                                 & 0.9179                                                               & 0.5722                                                             & 0.8515           \\ \cline{2-9} 
\multicolumn{1}{c|}{}                                     & \multirow{8}{*}{\textbf{VBench}}        & \multicolumn{1}{c}{\textbf{Method}} & \textbf{\begin{tabular}[c]{@{}c@{}}Subject\\ Consistency\end{tabular}} & \textbf{\begin{tabular}[c]{@{}c@{}}Background\\ Consistency\end{tabular}} & \textbf{\begin{tabular}[c]{@{}c@{}}Temporal\\ Flickering\end{tabular}} & \textbf{\begin{tabular}[c]{@{}c@{}}Motion\\ Smoothness\end{tabular}} & \textbf{\begin{tabular}[c]{@{}c@{}}Imaging\\ Quality\end{tabular}} & \textbf{Average} \\ \cline{3-9} 
\multicolumn{1}{c|}{}                                     &                                         & Centralized                         & 0.9833                                                                 & 0.9809                                                                    & 0.9831                                                                 & 0.9904                                                               & 0.7246                                                             & 0.9325           \\
\multicolumn{1}{c|}{}                                     &                                         & NMP                                 & 0.9839                                                                 & 0.9803                                                                    & 0.9832                                                                 & 0.9883                                                               & 0.7241                                                             & 0.9320           \\
\multicolumn{1}{c|}{}                                     &                                         & PP                                  & 0.9942                                                                 & 0.9879                                                                    & 0.9866                                                                 & 0.9906                                                               & 0.7353                                                             & 0.9389           \\
\multicolumn{1}{c|}{}                                     &                                         & HP                                  & 0.9878                                                                 & 0.9827                                                                    & 0.9838                                                                 & 0.9904                                                               & 0.7201                                                             & 0.9329           \\
\multicolumn{1}{c|}{}                                     &                                         & LP ($r$=1.0)                          & 0.9691                                                                 & 0.9733                                                                    & 0.9803                                                                 & 0.9888                                                               & 0.7550                                                             & 0.9333           \\
\multicolumn{1}{c|}{}                                     &                                         & LP ($r$=0.5)                          & 0.9606                                                                 & 0.9648                                                                    & 0.9787                                                                 & 0.9873                                                               & 0.7516                                                             & 0.9286           \\
\multicolumn{1}{c|}{}                                     &                                         & VideoCrafter                        & 0.8624                                                                 & 0.9288                                                                    & 0.9760                                                                 & 0.9179                                                               & 0.5722                                                             & 0.8515           \\ \hline
\multirow{24}{*}{\textbf{81 Frames}}                      & \multirow{8}{*}{\textbf{EvalCrafter}}   & \multicolumn{1}{c}{\textbf{Method}} & \textbf{\begin{tabular}[c]{@{}c@{}}Subject\\ Consistency\end{tabular}} & \textbf{\begin{tabular}[c]{@{}c@{}}Background\\ Consistency\end{tabular}} & \textbf{\begin{tabular}[c]{@{}c@{}}Temporal\\ Flickering\end{tabular}} & \textbf{\begin{tabular}[c]{@{}c@{}}Motion\\ Smoothness\end{tabular}} & \textbf{\begin{tabular}[c]{@{}c@{}}Imaging\\ Quality\end{tabular}} & \textbf{Average} \\ \cline{3-9} 
                                                          &                                         & Centralized                         & 0.9774                                                                 & 0.9803                                                                    & 0.9863                                                                 & 0.9921                                                               & 0.7202                                                             & 0.9313           \\
                                                          &                                         & NMP                                 & 0.9773                                                                 & 0.9856                                                                    & 0.9867                                                                 & 0.9924                                                               & 0.6792                                                             & 0.9242           \\
                                                          &                                         & PP                                  & 0.9776                                                                 & 0.9782                                                                    & 0.9825                                                                 & 0.9896                                                               & 0.7060                                                             & 0.9268           \\
                                                          &                                         & HP                                  & 0.9776                                                                 & 0.9824                                                                    & 0.9832                                                                 & 0.9912                                                               & 0.7170                                                             & 0.9303           \\
                                                          &                                         & LP ($r$=1.0)                          & 0.9663                                                                 & 0.9714                                                                    & 0.9834                                                                 & 0.9893                                                               & 0.7173                                                             & 0.9255           \\
                                                          &                                         & LP ($r$=0.5)                          & 0.9358                                                                 & 0.9606                                                                    & 0.9787                                                                 & 0.9879                                                               & 0.7036                                                             & 0.9133           \\
                                                          &                                         & VideoCrafter                        & 0.8624                                                                 & 0.9288                                                                    & 0.9760                                                                 & 0.9179                                                               & 0.5722                                                             & 0.8515           \\ \cline{2-9} 
                                                          & \multirow{8}{*}{\textbf{T2V-CompBench}} & \multicolumn{1}{c}{\textbf{Method}} & \textbf{\begin{tabular}[c]{@{}c@{}}Subject\\ Consistency\end{tabular}} & \textbf{\begin{tabular}[c]{@{}c@{}}Background\\ Consistency\end{tabular}} & \textbf{\begin{tabular}[c]{@{}c@{}}Temporal\\ Flickering\end{tabular}} & \textbf{\begin{tabular}[c]{@{}c@{}}Motion\\ Smoothness\end{tabular}} & \textbf{\begin{tabular}[c]{@{}c@{}}Imaging\\ Quality\end{tabular}} & \textbf{Average} \\ \cline{3-9} 
                                                          &                                         & Centralized                         & 0.9707                                                                 & 0.9768                                                                    & 0.9807                                                                 & 0.9874                                                               & 0.6637                                                             & 0.9159           \\
                                                          &                                         & NMP                                 & 0.9681                                                                 & 0.9809                                                                    & 0.9818                                                                 & 0.9885                                                               & 0.6779                                                             & 0.9194           \\
                                                          &                                         & PP                                  & 0.9427                                                                 & 0.9690                                                                    & 0.9772                                                                 & 0.9872                                                               & 0.7043                                                             & 0.9161           \\
                                                          &                                         & HP                                  & 0.9738                                                                 & 0.9756                                                                    & 0.9833                                                                 & 0.9894                                                               & 0.6978                                                             & 0.9240           \\
                                                          &                                         & LP ($r$=1.0)                          & 0.9556                                                                 & 0.9642                                                                    & 0.9780                                                                 & 0.9850                                                               & 0.7103                                                             & 0.9186           \\
                                                          &                                         & LP ($r$=0.5)                          & 0.9409                                                                 & 0.9576                                                                    & 0.9771                                                                 & 0.9845                                                               & 0.6628                                                             & 0.9045           \\
                                                          &                                         & VideoCrafter                        & 0.8624                                                                 & 0.9288                                                                    & 0.9760                                                                 & 0.9179                                                               & 0.5722                                                             & 0.8515           \\ \cline{2-9} 
                                                          & \multirow{8}{*}{\textbf{VBench}}        & \multicolumn{1}{c}{\textbf{Method}} & \textbf{\begin{tabular}[c]{@{}c@{}}Subject\\ Consistency\end{tabular}} & \textbf{\begin{tabular}[c]{@{}c@{}}Background\\ Consistency\end{tabular}} & \textbf{\begin{tabular}[c]{@{}c@{}}Temporal\\ Flickering\end{tabular}} & \textbf{\begin{tabular}[c]{@{}c@{}}Motion\\ Smoothness\end{tabular}} & \textbf{\begin{tabular}[c]{@{}c@{}}Imaging\\ Quality\end{tabular}} & \textbf{Average} \\ \cline{3-9} 
                                                          &                                         & Centralized                         & 0.9557                                                                 & 0.9729                                                                    & 0.9804                                                                 & 0.9891                                                               & 0.7552                                                             & 0.9307           \\
                                                          &                                         & NMP                                 & 0.9663                                                                 & 0.9714                                                                    & 0.9802                                                                 & 0.9883                                                               & 0.7329                                                             & 0.9278           \\
                                                          &                                         & PP                                  & 0.9770                                                                 & 0.9730                                                                    & 0.9837                                                                 & 0.9904                                                               & 0.7255                                                             & 0.9299           \\
                                                          &                                         & HP                                  & 0.9800                                                                 & 0.9711                                                                    & 0.9885                                                                 & 0.9926                                                               & 0.6951                                                             & 0.9255           \\
                                                          &                                         & LP ($r$=1.0)                          & 0.9688                                                                 & 0.9661                                                                    & 0.9757                                                                 & 0.9870                                                               & 0.7580                                                             & 0.9311           \\
                                                          &                                         & LP ($r$=0.5)                          & 0.9553                                                                 & 0.9476                                                                    & 0.9760                                                                 & 0.9842                                                               & 0.7365                                                             & 0.9199           \\
                                                          &                                         & VideoCrafter                        & 0.8624                                                                 & 0.9288                                                                    & 0.9760                                                                 & 0.9179                                                               & 0.5722                                                             & 0.8515           \\ \hline
\end{tabular}

\label{detailed-evaluation}
\end{table*}

\begin{figure*}[!t]
    \centering
    \begin{minipage}{\linewidth}
        \centering
        \includegraphics[width=0.98\textwidth]{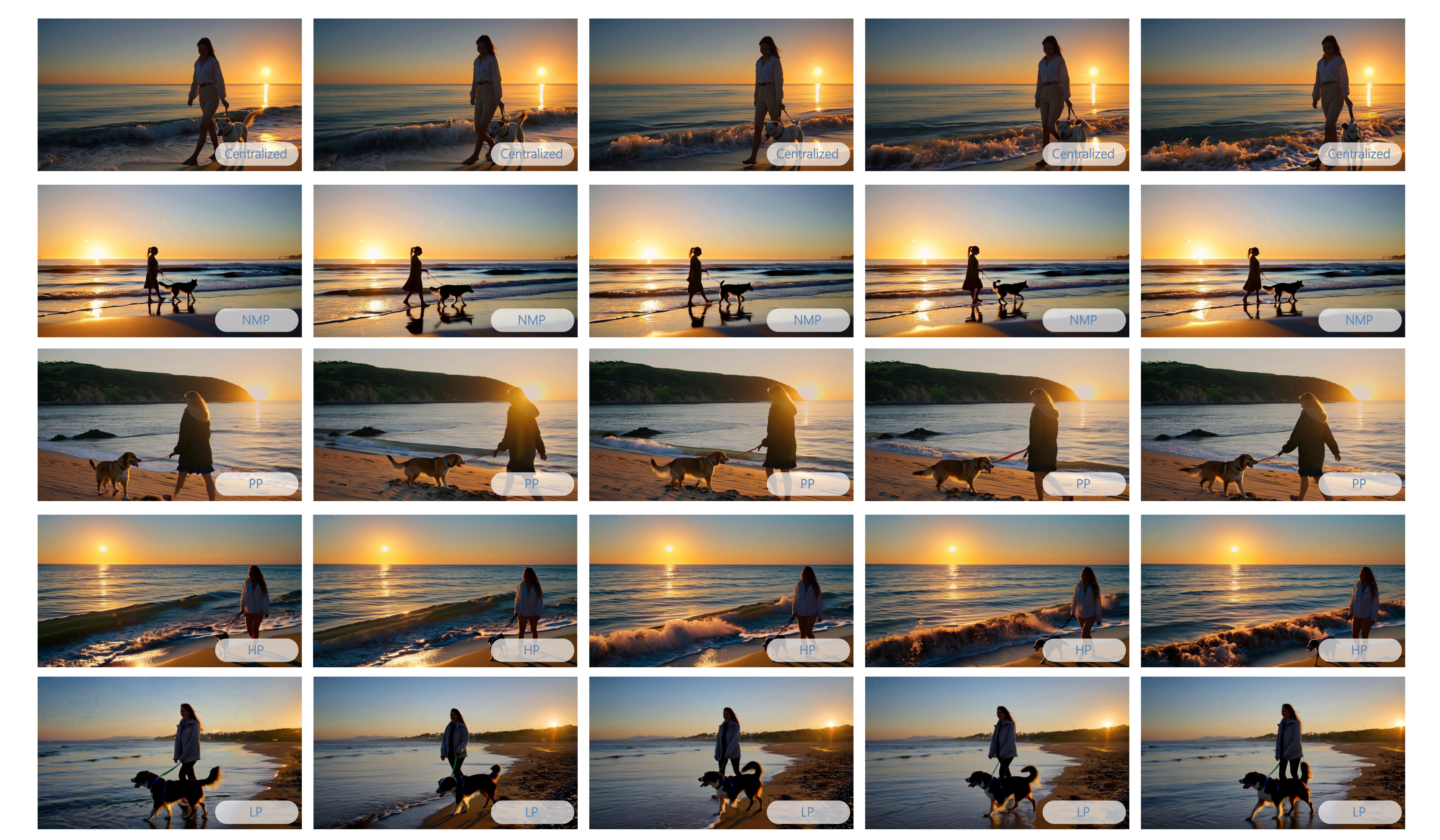}
        \caption{Videos generated using a prompt sampled from EvalCrafter. Prompt: "A woman is walking her dog on the beach at sunset."}
        \label{evalcrafter-vis}
    \end{minipage}\par\medskip
\vspace{5pt}
    \centering
    \begin{minipage}{\linewidth}
        \centering
        \includegraphics[width=0.98\textwidth]{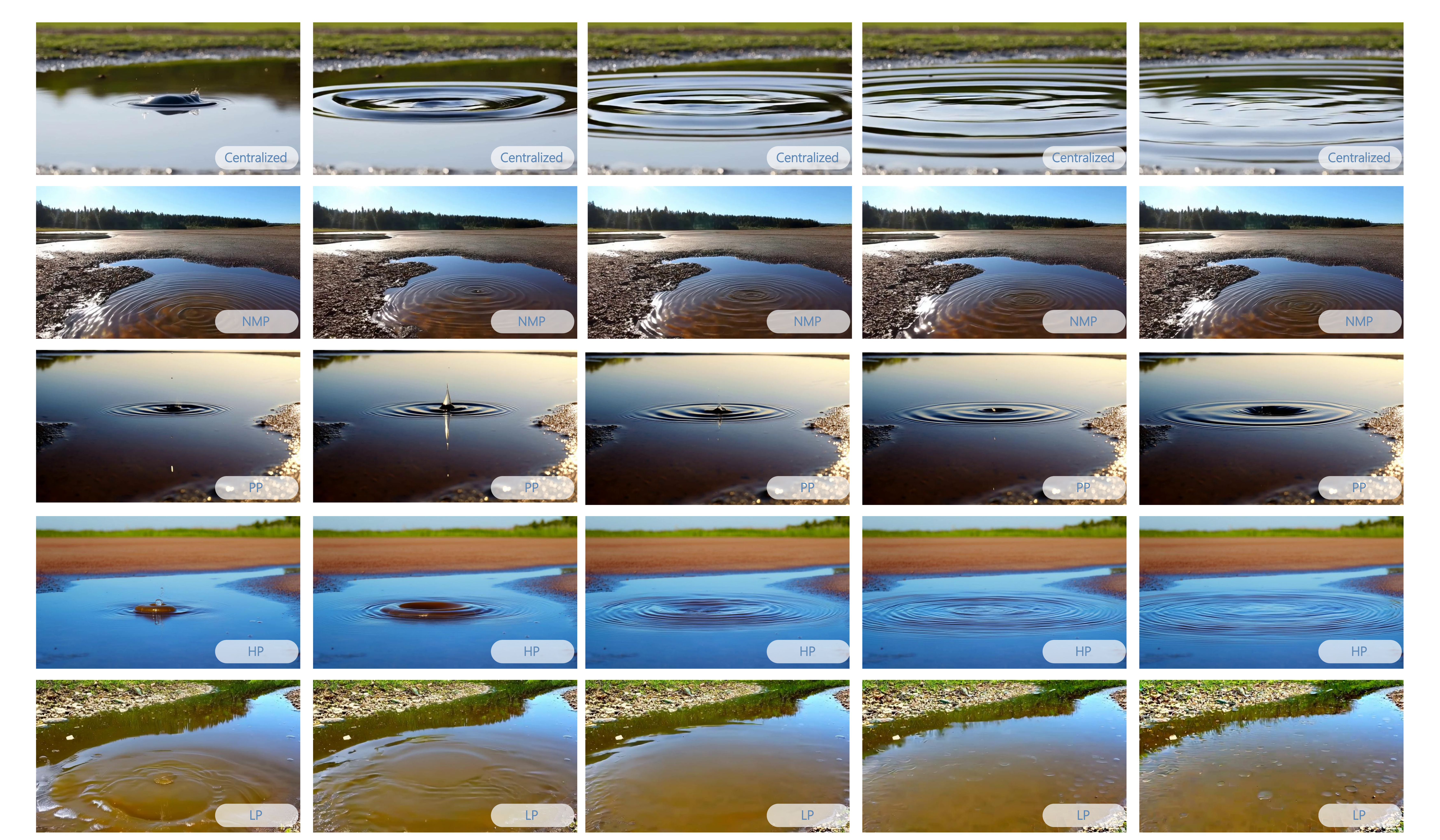}
        \caption{Videos generated using a prompt sampled from T2V-CompBench. Prompt: "A timelapse of a puddle drying up and disappearing on a hot day."}
        \label{t2v-vis}
    \end{minipage}\par\medskip
\end{figure*}

\begin{figure*}[!t]
    \centering
    \begin{minipage}{\linewidth}
        \centering
        \includegraphics[width=0.98\textwidth]{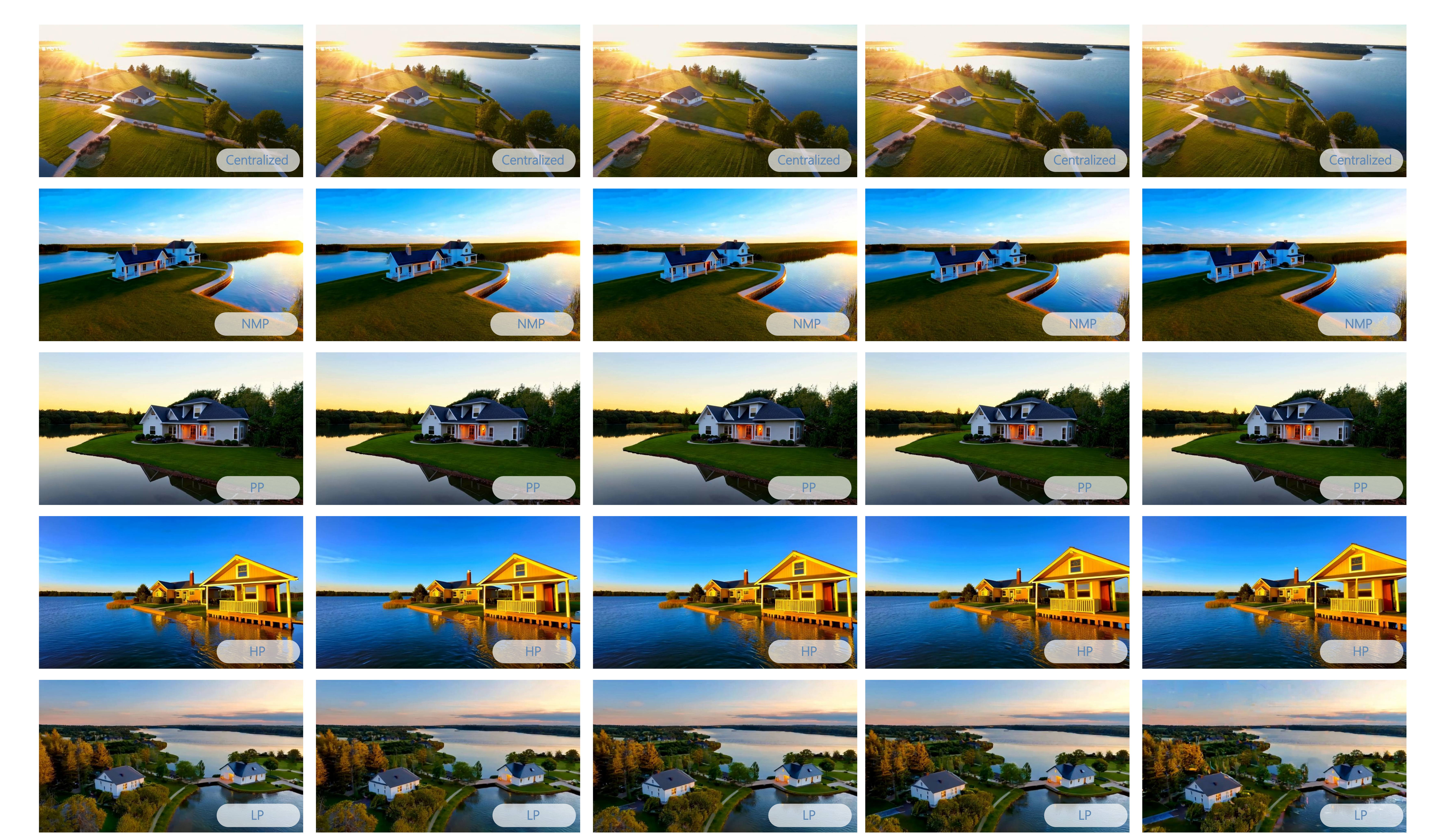}
        \caption{Videos generated using a prompt sampled from VBench. Prompt: "Drone view of house near lake during golden hour."}
        \label{vbench-vis}
    \end{minipage}\par\medskip
\end{figure*}

\section{Detailed Quantitative Evaluation Results}
In this section, we provide the detailed quantitative evaluation results that serve as the basis for the radar charts presented in Figure \label{radar}. Due to space constraints in the main body, the complete numerical results are presented here to allow for a precise comparison.
Table \ref{detailed-evaluation} provides a detailed breakdown of the scores for each quality metric across the EvalCrafter, T2V-CompBench, and VBench benchmarks. It compares our LP method against other baseline parallelism strategies (NMP, PP, HP), the centralized generation mode, and VideoCrafter for both 49-frame and 81-frame video generation. The metrics evaluated are Subject Consistency, Background Consistency, Temporal Flickering, Motion Smoothness, and Imaging Quality.
A detailed analysis of Table \ref{detailed-evaluation} shows that across all three benchmarks and for both 49- and 81-frame generation, the performance of LP (using either $r=1.0$ or $r=0.5$) consistently aligns with the high-performance cluster formed by other parallelism strategies (NMP, PP, HP) and the Centralized baseline. Crucially, the average quality scores between our LP ($r=1.0$) and the Centralized mode do not differ by more than 0.6\% in any of the tested scenarios. This high-fidelity cluster stands in sharp contrast to the VideoCrafter baseline, which exhibits consistently lower scores across all metrics. 
These results offer more robust numerical evidence supporting our claim that LP achieves substantial reductions in communication overhead with little compromise on video generation quality.

\section{Additional Qualitative Visualization Results}
To intuitively demonstrate the video generation quality of LP strategy, we provide a direct visual comparison between video frames generated by LP, the ideal Centralized mode, and other baseline parallelism strategies (NMP, PP, HP).

As illustrated in Figure \ref{evalcrafter-vis}, \ref{t2v-vis}, and \ref{vbench-vis}, the video frames generated by our LP method are visually indistinguishable from the Centralized mode and other parallelism strategies.
In Figure \ref{evalcrafter-vis}, all methods are compared on the same prompt sampled from EvalCrafter dataset. The LP-generated video maintains perfect subject consistency (the woman and dog), background coherence (the sunset and waves), and temporal smoothness. It flawlessly matches the video generated by the Centralized mode in every observable detail. Key elements, such as the sun's reflection on the wet sand, the natural gait of the dog, and the fine textures of the ocean waves, are all rendered with high fidelity and remain perfectly stable across frames.
Similarly, in Figure \ref{t2v-vis} and Figure \ref{vbench-vis}, the outputs from LP exhibit visual quality comparable with baselines. Figure \ref{t2v-vis} shows the natural, concentric expansion of ripples in water, with realistic reflections and lighting. Figure \ref{vbench-vis} demonstrates a smooth, stable drone shot, capturing the complex geometry of the houses and the landscape during golden hour. In all cases, the visual quality, lighting, and texture details produced by LP are on par with the high-fidelity Centralized results. These qualitative results strongly support our quantitative findings. They confirm that the LP strategy, despite its significant reduction in communication overhead, does not introduce perceptible visual artifacts or quality degradation.

\section{Combining LP with Conventional Parallelism Strategies}
\label{sec:lp_hybrid}
This section describes how the proposed LP can be integrated with existing strategies (e.g., NMP, TP, PP) to construct a hierarchical hybrid parallelism framework for large-scale GPU clusters. The workflow comprises two parts:
\begin{itemize}
  \item Inter-group LP. At the cluster level, all GPUs are first partitioned into groups, and LP is applied across groups by partitioning the latent space.
  \item Intra-group parallelism. Within each group, the corresponding sub-latent is treated as an independent video generation subtask, allowing any model parallelism strategy to be deployed for accelerating the serving process.
\end{itemize}
We will elaborate on these two parts below, followed by a communication 
overhead analysis of the complete framework.

\subsection{Inter-group LP}
Assume the cluster has in total $K$ GPUs formulated as follows:
\begin{equation}
  \mathcal{G} = \{1,2,\dots,K\}, \qquad |\mathcal{G}| = K.
\end{equation}
We partition $\mathcal{G}$ into $M$ disjoint GPU groups, subject to the following constrains:
\begin{equation}
  \begin{cases}
    \mathcal{G}_1, \mathcal{G}_2, \dots, \mathcal{G}_M \subset \mathcal{G}, \\
    \mathcal{G}_m \cap \mathcal{G}_{m'} = \varnothing, \quad \forall m, m' \in \{1,2,\dots,M\} \land m\neq m', \\
    \bigcup_{m=1}^{M}\mathcal{G}_m = \mathcal{G}, \\
    K_m = |\mathcal{G}_m|, \qquad \sum_{m=1}^{M} K_m = K.
  \end{cases}
\end{equation}
In hierarchical hybrid parallelism, we first perform inter-group LP along dimension $d_i$ on $z_t$. Specifically, we define an inter-group latent partition operator $\mathcal{P}^{\text{inter}}_{d_i}$ that maps the full latent into $M$ sub-latents for transfer across GPU groups, where $z_t^{[m]}$ denotes the sub-latent assigned to group $m$. The construction of these sub-latents follows the same patch-aligned overlapping partition strategy described Eqs. (\ref{partition-1}, \ref{partition-2}, \ref{partition-3}, \ref{partition-4}), except that the number of partitions $K$ is replaced by $M$. After intra-group parallel denoising, the position-aware latent reconstruction technique also follows the same principle as described in Eqs. (\ref{reconstruct-1}, \ref{reconstruct-2}, \ref{reconstruct-3}, \ref{reconstruct-4}, \ref{reconstruct-4-5}, \ref{reconstruct-5}, \ref{reconstruct-6}), again replacing the number of partitions from $K$ to $M$.

\subsection{Intra-Group Parallelism}
For each GPU group $\mathcal{G}_m$, we abstract the intra-group parallel execution as the following operator:
\begin{equation}
  \Phi_m : \big(z_t^{[m]}, c, t\big) \;\longmapsto\;
  \tilde{f}^{[m]}(z_t^{[m]}, t, c),
\end{equation}
where $\tilde{f}^{[m]}$ is the local noise prediction function on $z_t^{[m]}$ realized by any chosen model parallelism strategy such as NMP, TP and PP. Therefore, from the perspective of inter-group LP, $\Phi_m$ can be treated as a black-box operator that performs parallel denoising of the sub-latent $z_t^{[m]}$ across $K_m$ GPUs.

\subsection{Communication Overhead of LP-Integrated Hybrid Parallelism}
In LP-integrated hybrid parallelism, the communication overhead at each diffusion timestep $t$ can be decomposed into the following two parts:

\begin{itemize}
  \item Inter-group communication overhead, which comes from scattering and gathering sub-latents across GPU groups;
  \item Intra-group communication overhead, which is jointly determined by the number of GPU groups $M$ and the intra-group parallelism strategy.
\end{itemize}
Thus, the per-step total communication overhead is as follows:
\begin{equation}
  C^{\text{hyb}}_{\text{step}} 
  = C^{\text{inter}}_{\text{step}} + \sum_{m=1}^{M} C^{\text{intra}}_{m,\text{step}},
\end{equation}
and with $T$ diffusion steps, the overall communication overhead can be formulated as:
\begin{equation}
\begin{aligned}
 & \;\;\;\; C^{\text{hyb}}_{\text{total}} \\
&  = T \cdot C^{\text{inter}}_{\text{step}} 
    + \sum_{m=1}^{M} T \cdot C^{\text{intra}}_{m,\text{step}} \\
 & = C^{\text{inter}}_{\text{total}} + \sum_{m=1}^{M} C^{\text{intra}}_{m,\text{total}}.
\end{aligned}
\end{equation}

\noindent
\textbf{Inter-group communication overhead.} The per-step inter-group communication overhead $C_\text{step}^\text{inter}$ is analogous to the overhead of the standard LP strategy analyzed in section \ref{comm-cost-lp}, but applied across $M$ GPU groups instead of $K$ individual GPUs. At each timestep $t$, the master group (e.g., $\mathcal{G}_1$) scatters $M-1$ sub-latents $\{z_t^{[m]}\}_{m=2}^M$ to the other groups. After intra-group parallel denoising, it then gathers the $M-1$ local noise predictions for position-aware latent reconstruction. Let $S_{sub}^{[m]}$ be the size of the sub-latent assigned to group $\mathcal{G}_m$. Following the derivation in Eq. (\ref{step-lp}), the per-step inter-group overhead is:
\begin{equation}
\begin{aligned}
    & \;\;\;\;
    C^{\text{inter}}_{\text{step}} \\
    & = (C^\text{inter}_\text{Scatter} + C^\text{inter}_\text{Gather}) \times 2 \\
    &= \left( \sum_{m=2}^{M} S_{sub}^{[m]} + \sum_{m=2}^{M} S_{sub}^{[m]} \right) \times 2 \\
    &= 4 \sum_{m=2}^{M} S_{sub}^{[m]},
\end{aligned}
\end{equation}
and the total inter-group communication overhead across $T$ steps is:
\begin{equation}
\begin{aligned}
    & \;\;\;\;
    C_\text{total}^\text{inter} \\
    & = T \cdot C_\text{step}^\text{inter} \\
    & = 4T \sum_{m=2}^{M} S_{sub}^{[m]}.
\end{aligned}
\end{equation}

\noindent
\textbf{Intra-group communication overhead.} The per-step intra-group communication overhead $C_{m,\text{step}}^\text{intra}$ is determined by the specific parallelism strategy deployed within a given GPU group $\mathcal{G}_m$. This strategy is used to execute the denoising operator $\Phi_m$ on the sub-latent $z_t^{[m]}$ across $K_m$ GPUs within group $m$. Taking NMP as an example, the per-step intra-group communication overhead for a group can be computed as:
\begin{equation}
    C_{m,\text{step}}^\text{intra} = 2(K_m - 1) \cdot S'_H,
\end{equation}
following the deviation from Eq. (\ref{comm-nmp}). Here, $S'_H$ represents the size of the transferred activation tensor corresponding to the sub-latent $z_t^{[m]}$, which is consistently smaller than $S_H$ (i.e., $S'_H<S_H$).
As a result, the total intra-group communication overhead for group $m$ is:
\begin{equation}
    C_{m,\text{total}}^\text{intra} = T \cdot C_{m,\text{step}}^\text{intra} = 2T(K_m - 1) \cdot S'_H.
\end{equation}
The total communication overhead for the hybrid framework is the sum of the inter-group communication overhead (from LP) and the cumulative intra-group communication overhead from all $M$ groups. Using NMP as the intra-group parallelism strategy, the total overhead is:
\begin{equation}
\begin{aligned}
    & \;\;\;\;C_\text{total}^\text{hyb} \\&= C_\text{total}^\text{inter} + \sum_{m=1}^{M}C_{m,\text{total}}^\text{intra} \\&= \left( 4T\sum_{m=2}^{M}S_{sub}^{[m]} \right) + \left( \sum_{m=1}^{M} 2T(K_{m}-1)S'_{H} \right)\\
    & =2T \sum_{m=1}^{M} \left( 2S_{sub}^{[m]} + (K_{m}-1)S'_{H} \right)-4TS_{sub}^{[1]}.
\end{aligned}
\end{equation}
As $S_{sub}^{[m]}$ is empirically negligible compared with $S'_{H}$, we can therefore approximate the total overhead as:
\begin{equation}
\begin{aligned}
    & \;\;\;\;C_\text{total}^\text{hyb} \\
    &=2T \sum_{m=1}^{M} \left( 2S_{sub}^{[m]} + (K_{m}-1)S'_{H} \right)-4TS_{sub}^{[1]}\\
    &\approx2T \sum_{m=1}^{M} \left((K_{m}-1)S'_{H} \right)\\
    &\approx 2TS'_{H} \sum_{m=1}^{M} (K_{m}-1)\\
    &\approx 2TS'_{H} \left( \left( \sum_{m=1}^{M} K_{m} \right) - M \right).
\end{aligned}
\end{equation}
Since 
\begin{equation}
    \sum_{m=1}^{M} K_{m} = K,
\end{equation}
we have:
\begin{equation}
\begin{aligned}
    & \;\;\;\;C_\text{total}^\text{hyb} \\
    &\approx 2TS'_{H} \left( \left( \sum_{m=1}^{M} K_{m} \right) - M\right)\\
    & \approx 2TS'_{H} (K- M)
\end{aligned}
\end{equation}
We can now compare this to the overhead of a pure NMP derived in Eq. (\ref{nmp-comm-total}). Then, the ratio of their communication overheads is therefore:
\begin{equation}
\begin{aligned}
&\;\;\;\; \frac{C_\text{total}^\text{hyb}}{C_{NMP}} \\&\approx \frac{2TS'_{H} (K - M)}{2TS_{H}(K-1)} \\&< \frac{K - M}{K-1}
\end{aligned}
\end{equation}
This result demonstrates the significant advantage of the hierarchical hybrid approach. By partitioning the $K$ GPUs into $M$ independent groups, the dominant activation-transfer overhead is reduced by at least a factor of $\frac{K - M}{K-1}$. Our preliminary experiments confirm this theoretical upper bound on communication overhead reduction. This reduction occurs because the NMP communication chain is broken into $M$ shorter, parallel chains, and the links between these chains are handled by the highly efficient, low-overhead LP. This allows the hybrid inference system to scale to larger GPU clusters while effectively mitigating the primary communication bottleneck.

\end{document}